\documentclass{article}
\usepackage{ijcai25}

\usepackage{times}
\usepackage{soul}
\usepackage{url}
\usepackage[hidelinks]{hyperref}
\usepackage[utf8]{inputenc}
\usepackage[small]{caption}
\usepackage{graphicx}
\usepackage{amsmath}
\usepackage{amsthm}
\usepackage{booktabs}
\usepackage{algorithm}
\usepackage[noend]{algpseudocode}
\usepackage[switch]{lineno}
\usepackage{multirow}

\newtheorem{theorem}{Theorem}

\newtheorem{definition}{Definition}
\newtheorem{proposition}{Proposition}

\usepackage{amsfonts}
\usepackage{amsmath}
\usepackage{optidef}
\usepackage{tikz,pgffor}
\usetikzlibrary{arrows,backgrounds,calc}
\usetikzlibrary{automata,positioning,arrows,shapes,math,arrows.meta,decorations.pathmorphing}
\tikzstyle{min}=[thick,circle,draw,minimum size=1.3em,inner sep=0em,text centered]

\newcommand{\supp}{the Appendix}
\usepackage{enumitem}
\setlist{nosep}

\newcommand{\Nset}{\mathbb{N}}

\newcommand{\Inv}{\mathbb{I}}
\newcommand{\Out}{\textit{Out}}
\newcommand{\In}{\textit{In}}

\newcommand{\Run}{\mathrm{Run}}
\newcommand{\MRun}{\mathrm{MRun}}

\newcommand{\Col}{\textit{Col}}
\newcommand{\cols}{\gamma}
\newcommand{\prob}{\mathbb{P}}
\newcommand{\Prob}{\textit{P}}

\newcommand{\Freq}{\mathit{Freq}}

\newcommand{\Obj}{\textit{Obj}}

\newcommand{\Dist}{\mathit{Dist}}
\newcommand{\dist}{\mathcal{D}}
\newcommand{\AltDist}{\mathit{AltDist}}

\newcommand{\Ach}{\mathcal{F}}
\newcommand{\HR}{\textit{HR}}
\newcommand{\FR}{\textit{FR}}
\newcommand{\MR}{\textit{MR}}
\newcommand{\FMR}{\textit{FMR}}

\newcommand{\calA}{\mathcal{A}}

\newcommand{\PSPACE}{$\mathsf{PSPACE}$}

\newcommand{\NP}{$\mathsf{NP}$}
\newcommand{\coNP}{$\mathsf{coNP}$}

\usepackage{xspace}
\makeatletter
\DeclareRobustCommand\onedot{\futurelet\@let@token\@onedot}
\def\@onedot{\ifx\@let@token.\else.\null\fi\xspace}

\makeatother

\title{Steady-State Strategy Synthesis for Swarms of Autonomous Agents}

\author{
Martin Jon\'{a}\v{s}
\and
Anton\'{\i}n Ku\v{c}era
\and
Vojt\v{e}ch K\r{u}r\And
Jan Ma\v{c}\'{a}k\\
\affiliations
Faculty of Informatics, Masaryk University, Czechia
\emails
martin.jonas@mail.muni.cz, tony@fi.muni.cz, vojtech.kur@mail.muni.cz, 
macak.jan@mail.muni.cz
}

\begin{document}

\maketitle

\begin{abstract}
    Steady-state synthesis aims to construct a policy for a given MDP $D$ such that the long-run average frequencies of visits to the vertices of $D$ satisfy given numerical constraints. This problem is solvable in polynomial time, and memoryless policies are sufficient for approximating an arbitrary frequency vector achievable by a general (infinite-memory) policy.

    We study the steady-state synthesis problem for \emph{multiagent systems}, where multiple autonomous agents jointly strive to achieve a suitable frequency vector. We show that the problem for multiple agents is computationally hard (\PSPACE\ or \NP\ hard, depending on the variant), and 
    memoryless strategy profiles are \emph{insufficient} for approximating achievable frequency vectors. Furthermore, we prove that even \emph{evaluating} the frequency vector achieved by a given memoryless profile is computationally hard. This reveals a severe barrier to constructing an efficient synthesis algorithm, even for memoryless profiles. Nevertheless, we design an \emph{efficient and scalable} synthesis algorithm for a subclass of \emph{full} memoryless profiles, and we evaluate this algorithm on a large class of randomly generated instances. The experimental results demonstrate a significant improvement against a naive algorithm based on strategy sharing.
\end{abstract}

\section{Introduction}
\label{sec-intro}

Steady-state policy synthesis is the problem of computing a suitable decision-making policy (strategy) in a given Markov decision process (MDP) $D$ satisfying given constraints on the limit frequencies of visits to the states of~$D$. More precisely, we say that a strategy $\sigma$ in $D$ \emph{achieves} a frequency vector $\mu$ if for almost every infinite run $w$ in the Markov chain $D^\sigma$ induced by the strategy and every vertex $v$ of $D$ we have that the limit frequency of visits to $v$ along $w$ is equal to $\mu(v)$.

% A standard result of ergodic theory (see, e.g., \cite{Norris:book} says that the invariant distribution is the vector of limit frequencies of visits to the states of $D$, and this frequency vector is the same for almost all runs of $D$ induced by the strategy. Hence, the steady-state policy synthesis is a fundamental tool for optimizing the long-run average behaviour of MDPs. 
%
%One important generalization of this problem is the ``colored'' steady-state policy synthesis where the states are first split into disjoint subsets (colors), and the frequency constraints are then formulated for the colors.  
The existing works concentrate on the steady-state synthesis problem for a \emph{single} agent, where the task is to construct a strategy $\sigma$ achieving a frequency vector $\mu$ where $\vec{v}_\ell \leq \mu \leq \vec{v}_u$ for given lower and upper bounds $\vec{v}_\ell$ and~$\vec{v}_u$. The existence of such a strategy is decidable in polynomial time; and if it exists, it can be also computed in polynomial time by linear programming (see \emph{Related work}). Although some frequency vectors are only achievable by infinite-memory strategies, the subclass of \emph{memoryless} strategies is sufficient for producing frequency vectors \emph{arbitrarily close} to each achievable frequency vector. If the underlying graph of $D$ is strongly connected, then the same holds even for a special type of \emph{full} memoryless strategies assigning a positive probability to every edge of~$D$ (one can show that for every $\mu$ achievable by a memoryless strategy, there is a vector arbitrarily close to $\mu$ achievable by a full memoryless strategy). These properties are illustrated in Fig.\ref{fig-one-agent} on a trivial MDP with three non-deterministic states. Hence, in the single agent setting, memoryless strategies are sufficient for practical applications. Since we can safely assume that $D$ is a disjoint union of finitely many strongly connected MDPs, the same holds even for \emph{full} memoryless strategies (see Section~\ref{sec-algorithm} for details).
% We refer to \emph{Related work} for more details.

\paragraph{Our contribution}
In this paper, we extend the scope of steady-state policy synthesis problem to \emph{multiple autonomous agents}. More precisely, the task is to construct a strategy profile for $k{\geq} 1$ agents in a given MDP $D$ so that the frequencies of visits to the vertices of $D$ (or, more generally, to pre-defined classes of vertices represented by \emph{colors}) by some agent are above a given threshold vector. As a simple example, consider an MDP where the vertices represent devices requiring regular maintenance, and the threshold frequency vector specifies the minimal required frequency of inspections for each device. The steady-state policy synthesis for~$k$~agents then corresponds to the problem of designing appropriate schedules for $k$~independent technicians such that the required frequency of inspections is observed.\footnote{If two or more technicians meet at the same vertex at the same time, only one of them does the maintenance job. Hence, optimal strategy profiles tend to minimize the frequency of such redundant simultaneous visits. However, this redundancy cannot be avoided completely in general.} Our main results are twofold.
\smallskip

\noindent
\emph{\bfseries I.~Fundamental properties of the problem.} We analyze the role of memory and randomization in constructing (sub)optimal strategy profiles, and we also classify the computational complexity of the steady-state policy synthesis problem. The obtained results demonstrate that the steady-state policy synthesis for multiple agents is (perhaps even surprisingly) \emph{more complex} than for a single agent. Consequently, a different algorithmic approach is required. More concretely, we prove the following:
%\smallskip

\emph{I(a). For two or more agents, the power of full memoryless, memoryless, finite-memory, and general strategy profiles increases strictly.} To explain this, we need to introduce one extra notion. Let $k \geq 1$, and let $A,B$ be sets of strategy profiles for an MDP $D$ and $k$~agents. Furthermore, let $\Ach(A)$ and $\Ach(B)$ be the sets of frequency vectors achievable by the profiles in $A$ and $B$.
We say that $B$ is \emph{more powerful} than $A$ if $\Ach(A) \subseteq \Ach(B)$, and there exists $\mu \in \Ach(B)$ that cannot be approximated by the vectors of $\Ach(A)$ (i.e., there is $\delta > 0$ such that the distance between $\mu$ and every $\nu \in \Ach(A)$ is at least $\delta$).

We show that for $k \geq 2$, the sets of full memoryless profiles, memoryless profiles, finite-memory profiles, and general profiles are increasingly more powerful
even for strongly connected \emph{graphs}, i.e., MDPs without stochastic vertices. 
%We refer to Section~\ref{sec-fundaments} for details. %% JANEK: Toto je zavádějící formulace, jelikož s narůstajícím počtem paměťových elementů je možné libovolně přesně aproximovat obecné profily. Platí však, že existují příklady, kde žádný fixní počet paměťových elementů nestačí k tomu, aby bylo lze libovolně přesně aproximovat obecné profily.
This contrasts sharply with the single agent scenario where full memoryless profiles approximate general profiles.
%\smallskip

\emph{I(b) The existence of an achievable $\mu$ such that $\mu \geq \vec{v}_\ell$ for a given $\vec{v}_\ell$ is a computationally hard problem.} %% JANEK: Jednak zvláštní formulace (že existence je těžká), jednak pro fixní počet agentů a obecné profily je problém řešitelný v polynomiálním čase („rozvinutím“ stavového prostoru a řešením jakoby pro jednoho agenta).
Recall that for a single agent, the problem is solvable in polynomial time. For two or more agents, the problem is \NP-hard even if $D$ is a strongly connected graph and the set of profiles is restricted to full memoryless profiles, memoryless profiles, or finite-memory profiles with~$m$ memory states. For the ``colored'' variant of the problem, we obtain even \PSPACE-hardness.
%\smallskip

\emph{I(c) Evaluating the frequency vector achieved by a given profile is computationally hard, even for strongly connected graphs and memoryless profiles.} Intuitively, the reason is that each strategy in the profile may induce a Markov chain with a different period. The complexity of the evaluation procedure depends on the least common multiple of these periods whose size can be exponential in~$k$. Note that for \emph{full} memoryless profiles, all of the induced Markov chains have the \emph{same} period. Consequently, full memoryless profiles can be evaluated in polynomial time on strongly connected MDPs. These observations have important algorithmic consequences explained in the subsection \emph{II. Efficient synthesis algorithm}. 
%\smallskip

\emph{I(d) The existence of a finite-memory profile with $m$ memory states achieving a frequency vector $\mu$ such that $\mu \geq \vec{v}_\ell$ for a given $\vec{v}_\ell$ is decidable in polynomial space for every fixed number of agents.
% For full memoryless profiles, the problem is decidable in polynomial space.
} This holds also for the ``colored'' variant of the problem. The algorithm is based on encoding the problem as a formula of first order theory of the reals and applying the results of \cite{Canny:Tarski-exist-PSPACE}. The size of the formula is exponential in~$k$, which shows that the number of agents is a key parameter negatively influencing the computational costs.
% (this is not particularly surprising).
\smallskip

\noindent
\emph{\bfseries II. Efficient synthesis algorithm.} 
%The results of~I.{} establish the principal limitations of effective steady-state policy synthesis for $k\geq 2$ agents. 
Since general (infinite-memory) strategies are not algorithmically workable, the scope of algorithmic synthesis is naturally limited to \emph{finite-memory} profiles. 
The synthesis of a finite-memory profile for an MDP~$D$ where every strategy in the profile uses at most $m$~memory states is equivalent to the synthesis of a \emph{memoryless} profile for an MDP $D'$ obtained from $D$ by augmenting its vertices with memory states (see Section~\ref{sec-model} for details).   
%
% Suppose that every strategy in a finite-memory profile uses $m$ memory states. Then, the synthesis of a finite-memory profile for an MDP $D$ with the set of vertices $V$ is equivalent to the synthesis of a \emph{memoryless} profile for an MDP $\ag{D}$ where $V {\times} \{1,\ldots,m\}$ is the set of vertices and $(v,i) {\to} (u,j)$ iff $i,j \in \{1,\ldots,m\}$ and $v {\to} u$ is an edge of $D$ (see Section~\ref{sec} for details). %Since the number of vertices of $\ag{D}$ is linear in $m$, 
Hence, the algorithmic core of the problem is the construction of \emph{memoryless} profiles. However, here we face the obstacle of~I(c), saying that even \emph{evaluating} memoryless profiles is computationally hard. This is a severe barrier, because every synthesis algorithm is driven by the objective involving the frequency vector of the constructed profile. 
Hence, a natural starting point is to explore the constructability of \emph{full} memoryless profiles that can be evaluated in polynomial time (see I(c)). This is challenging, despite the limitations identified in~I(a). According to I(b), the associated decision problem is \NP-hard even for two agents, and the synthesis can be seen as a non-linear optimization problem whose size increases with the number of agents (see Section~\ref{sec-algorithm}).

%a straightforward use of mathematical programming leads to non-linear formulation of the problem   Furthermore, the synthesis algorithm needs to overcome the exponential blowup in the number of agents (see I(d)). %% JANEK: Pro full bezpaměťové profily momentálně nevidím důvod hovořit o nějakém exponenciálním nárůstu vzhledem k počtu agentů (nejmenší společný násobek period bude shora omezený počtem vrcholů grafu).

We propose an efficient algorithm for synthesizing full memoryless profiles based on \emph{incremental agent inclusion}. The main idea is the following: %% JANEK: zvážit vhodnost použití velkého písmena za dvojtečkou
Suppose that we already constructed a full memoryless profile for $k$~agents, and we wish to extend the profile to $k{+}1$ agents. Our algorithm constructs several linear programs depending only on the threshold vector (the objective) and numerical parameters extracted from the previously computed profile for $k$ agents. %% JANEK: Toto asi není v periodickém případě tak docela pravda: záleží tu nejen na celkovém frequency vectoru, nýbrž i na pravděpodobnostech návštěv v časových okamžicích příslušným jednotlivým zbytkovým třídám (modulo perioda), konkrétně záleží i na počátečním rozestavení předchozích agentů.
Hence, the size of these programs is \emph{independent of $k$}. A full memoryless strategy for the newly included agent is extracted from the solutions of these linear programs. Thus, we prevent the blowup in~$k$, and the complexity of our synthesis algorithm becomes \emph{linear} in the number of agents~$k$. 
Thus, we (inevitably) trade efficiency for completeness, i.e., the algorithm does not have to find a suitable full MR profile even if it exists.
We evaluate our algorithm experimentally on a series of randomly generated instances, and we show that it clearly outperforms a naive algorithm based on strategy sharing (see Section~\ref{sec-experiments} for details).

%To the best of our knowledge, this paper is the first attempt to solve an optimization problem for multiagent systems involving limit-average optimization criteria.

\paragraph{Related work}

All existing works about steady-state synthesis apply to a single agent scenario. \cite{ABHH:steady-state-ergodic-MDP} solve the problem for \emph{unichain MDPs}, i.e., a subclass of MDPs where every memoryless deterministic policy induces an ergodic Markov chain, by designing a polynomial-space algorithm. A polynomial-time algorithm for general MDPs is given in \cite{BBCFK:MDP-two-views-LMCS}. This algorithm can compute \emph{infinite-memory} strategies, which may be necessary for achieving some frequency vectors (see Fig.~\ref{fig-one-agent}), and it is applicable to a more general class of multiple mean-payoff objectives. It has been implemented \cite{BCFK:Multigain-TACAS} on top of the PRISM model checker \cite{KNP-PRISM4-CAV}. In \cite{Velasquez-steady-state}, the problem of constructing a suitable memoryless policy inducing a recurrent Markov chain consisting of all vertices of a given MDP is solved by linear programming. A generalization of this work is presented in \cite{ABAV:steady-state-synthesis-multichain}. Recent works \cite{Kretinsky:steady-state-LTL,VABTA:LTL-steady-state-AAMAS-journal} combine steady-state constraints with LTL specifications. There are also works concentrating on steady-state deterministic policy synthesis \cite{VASWA:deterministic-steady-state-JAR}.%, which is computationally harder than the synthesis of randomized policies.

\begin{figure}[t]\centering
    \begin{tikzpicture}[x=1.1cm, y=1.1cm, >=stealth',font=\footnotesize]
        \foreach \i/\x/\y in {0/0/0,1/4/0,2/2/-1.8}{
            \coordinate (c\i) at (\x,\y);
            \node[min] (A\i) at (c\i) {$v_1$};
            \node[min] (B\i) at ($(c\i) +(1,0)$) {$v_2$};
            \node[min] (C\i) at ($(c\i) +(2,0)$) {$v_3$};
            \node[below=.35 of A\i] 
              {\ifthenelse{\i=1}{$0.5 {-} \delta$}{} 
               \ifthenelse{\i=2}{$1{-}\delta_1{-}\delta_2$}{}};
            \node[below=.35 of B\i] 
              {\ifthenelse{\i=1}{$2\delta$}{}
               \ifthenelse{\i=2}{$\delta_1$}{}};            
            \node[below=.35 of C\i] 
              {\ifthenelse{\i=1}{$0.5 {-} \delta$}{}
               \ifthenelse{\i=2}{$\delta_2$}{}};
            \draw[->,thick,bend left] (A\i) to node[above] {\ifthenelse{\i=0}{}{$\varepsilon$}} (B\i);
            \draw[->,thick,bend left] (B\i) to node[below] 
              {\ifthenelse{\i=1}{$0.5$}{} 
               \ifthenelse{\i=2}{$1{-}\varepsilon$}{}}  (A\i);            
            \draw[<-,thick,bend left] (B\i) to node[above] 
              {\ifthenelse{\i=1}{$\varepsilon$}{}
               \ifthenelse{\i=2}{$1{-}\varepsilon$}{}} (C\i);
            \draw[<-,thick,bend left] (C\i) to node[below] {
                \ifthenelse{\i=1}{$0.5$}{}
                \ifthenelse{\i=2}{$\varepsilon$}{}} (B\i);
            \draw[->,thick,out=150,in=210,looseness=7] (A\i) to node[above=.15] {\ifthenelse{\i=0}{}{$1{-}\varepsilon$}{}} (A\i);
            \draw[->,thick,out=30,in=-30,looseness=7] (C\i) to node[above=.15] {\ifthenelse{\i=1}{$1{-}\varepsilon$}{}
             \ifthenelse{\i=2}{$\varepsilon$}{}} (C\i);
            }
\end{tikzpicture}
\caption{Left: A simple MDP with three non-deterministic vertices $v_1$, $v_2$, and $v_3$. %where the required frequency of visits is ${\geq} 0.5$, ${\geq} 0$, and ${\geq} 0.5$, respectively. 
Right: A memoryless strategy for a single agent can achieve the frequency vector $(0.5{-}\delta, 2\delta, 0.5{-}\delta)$ for an arbitrarily small $\delta{>}0$ by choosing a sufficiently small $\varepsilon {>} 0$.
However, the frequency vector $(0.5,0,0.5)$ is achievable only by an infinite-memory strategy where the $\varepsilon$ is ``progressively smaller'' and approaches $0$ as the vertices $v_1$ and $v_3$ are revisited. Middle: A full memoryless strategy can achieve the frequency vector $(1{-}\delta_1{-}\delta_2,\delta_1,\delta_2)$ where $\delta_1{+}\delta_2>0$ is arbitrarily small by choosing a sufficiently small $\varepsilon{>}0$. However, the vector $(1,0,0)$ is achievable only by a (non-full) strategy assigning~$1$ to the self-loop $v_1 {\to} v_1$.}    
\label{fig-one-agent}
\end{figure}
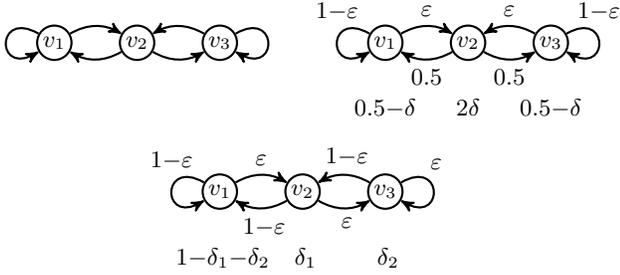
    
\section{The Model}
\label{sec-model}

We assume familiarity with basic notions of probability theory and Markov chain theory.
We use $\Nset$ and $\Nset_+$ to denote the sets of all non-negative and positive integers, respectively, and $\dist(A)$ to denote the set of all probability distributions over a finite set~$A$. A directed graph is a pair $G = (V,E)$ where $E \subseteq V \times V$. For every $v \in V$, we use $\In(v)$ and $\Out(v)$ to denote the sets of all in-going and out-going edges of $v$. We say that $G$ is \emph{strongly connected} if for all $v,u \in V$ there is a finite sequence $v_1,\ldots,v_n$ such that $n \geq 1$, $v_1 = v$, $v_n=u$, and $(v_i,v_{i+1}) \in E$ for all $1 \leq i <n$.
%For a given $a \in A$, we use $\mu[a]$ to denote the Dirac distribution assigning $1$ to $a$ and $0$ to the other elements of~$A$.

\paragraph{Markov chains}
A \emph{Markov chain} is a triple $C=(S,\Prob, \alpha)$ where $S$ is  a finite set of states, $\Prob \colon S\times S \to [0,1]$ is a stochastic matrix  such that $\sum_{s' \in S}\Prob(s,s') = 1$ for every $s \in S$, and $\alpha \in \dist(S)$ is an initial distribution. %% JANEK: Zvážil bych požadavek na neprázdnost množiny stavů S. Upozorňuji, že takto uvedená definice je ve skutečnosti definicí speciální podtřídy Markovových řetězců: konkrétně konečných diskrétněčasových (nikoli obecných).

A state $t$ is \emph{reachable} from a state $s$ if $\Prob^n(s,t) > 0$ for some $n\geq 1$, where $\Prob^n$ denotes the $n$-th power of~$\Prob$. A \emph{bottom strongly connected component (BSCC)} of $C$ is a maximal $B \subseteq S$ such that $B$ is strongly connected and closed under reachable states, i.e., for all $s,t \in B$ and $r \in S$ we have that $t$ is reachable from $s$, and if $r$ is reachable from $s$, then $r \in B$.
A Markov chain $C$ is \emph{irreducible} if for all $s,t \in S$ %such that $s \neq t$ %% JANEK: Navrhuji vypustit „such that $s \neq t$“.
we have that $t$ is reachable from $s$. 
We use $\Inv$ to denote the unique \emph{invariant distribution} of~$C$.
Note that every BSCC of $C$ can be seen as an irreducible Markov chain.

For every $s \in S$, let $d(s) = \gcd \{n \in \Nset_+ \mid \Prob^n(s,s) {>} 0\}$ be the \emph{period} of~$s$. Recall that if $C$ is irreducible, then $d(s)$ is the same for all $s \in S$ and defines the \emph{period} of $C$, denoted by $d$ (if $C$ is not clear, we write $d_C$ instead of $d$). Furthermore, the set $S$ can be partitioned into \emph{cyclic classes} $S_0,\ldots,S_{d-1}$ such that for all $i,j \in \{0,\ldots,d{-}1\}$ and $s,t \in S$ where $s \in S_i$ we have that $t \in S_j$ iff $\Prob^n(s,t) > 0$ for some $n \equiv (j{-}i) \ \mathrm{mod}\ d$. The structure of cyclic classes is shown in Fig.~\ref{fig-cyclic}. We say that $C$ is \emph{aperiodic} or \emph{periodic} depending on whether $d{=}1$ or not, respectively.

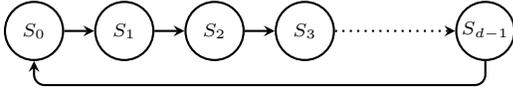
\begin{figure}[t]\centering
    \tikzstyle{st}=[thick,circle,draw,minimum size=2.2em,inner sep=0em,text centered]
    \tikzstyle{ar}=[thick,-stealth,rounded corners]
    \begin{tikzpicture}[font=\scriptsize,x=1.2cm, y=1.2cm]
    \foreach \i/\x in {0/0, 1/1, 2/2, 3/3, 5/d{-}1}{
        \node[st] (n\i) at (\i,1) {$S_{\x}$};
    }
    \foreach \i/\j in {0/1, 1/2, 2/3}{
        \draw[ar] (n\i) to (n\j);
    }
    \draw[ar,dotted] (n3) to (n5);
    \draw[ar] (n5) -- +(0,-.6) -| (n0);
    %    \draw[ar,dotted,bend left] (n3) to (n4);
    % \foreach \i/\l/\x in {0/90/0, 1/30/1, 2/-30/2, 3/-90/3, 4/-210/d{-}1}{
    %     \node[st] (n\i) at (\l:1) {$S_{\x}$};
    % }
    % \foreach \i/\j in {0/1, 1/2, 2/3, 4/0}{
    %      \draw[ar] (n\i) to (n\j);
    % }
    % \draw[ar,dotted,bend left] (n3) to (n4);
    \end{tikzpicture}
\caption{The structure of cyclic classes. For all states $s,t$ we have that $P(s,t) > 0$ only if $s \in S_i$ and $t \in S_{i{+}1 ~\mathrm{mod}~ d}$ for some $i<d$.}
\label{fig-cyclic}
\end{figure}

% A \emph{run} in $C$ is an infinite sequence $\omega = s_1,s_2,\ldots$ of states. We use $\omega(j)$ to denote the state $s_j$. Furthermore, we use
% %$\Run_C$ to denote the set of all runs in~$C$, and
% $\prob$ to denote the probability measure in the standard probability space over the set of all runs ($\prob$ is determined by $\Prob$ and~$\mu$; see, e.g., \cite{Norris:book}).
%, and we use $\Init(w)$ to denote the initial state of $w$ (\ie, $\Init(w) = s_0$).

% Let $s,t \in S$. We say that $t$ is \emph{reachable} from $s$ if the probability of visiting $t$ from $s$ is positive, i.e., \mbox{$\Prob^n(s,t) > 0$} for some $n\geq 0$ (recall that $\Prob^0$ is the identity matrix).

% For every strongly connected $C$, there exists a unique \emph{invariant distribution} $\Inv$ over the states such that $\Inv = \Inv \cdot \Prob$. By ergodic theorem \cite{Norris:book}, for every $s \in S$, the value of $\Inv(s)$ is the frequency of visits to~$s$ along a run in~$C$. More precisely, for almost all runs $w$ of $C$ we have that $\lim_{n\to \infty} \#_s(w_n)/n = \Inv(s)$, where $\#_s(w_n)$ is the number of occurrences of $s$ in the prefix of $w$ of length~$n$.

% For every run  $w = s_0,s_1,\ldots$, let $\Init(w)$ be the initial state $s_0$ of $\omega$. For a given set $T \subseteq S$ of target states, we use $\Steps_T(w)$ to denote the number of steps need to reach a state of $T$ along $w$, i.e., least $i$ such that $s_i \in T$ (if there is no such $i$, then $\Steps_T(w) = \infty$).

\paragraph{Markov decision processes (MDPs)} A \emph{Markov decision process (MDP)}\footnote{The adopted MDP definition is standard in the area of graph games. It is
equivalent to the ``classical'' definition of \cite{Puterman:book} but leads to simpler notation.} is a triple
$D {=} (V,E,p)$ where $V$ is a finite set of \emph{vertices} partitioned into subsets $(V_N,V_S)$ of \emph{non-deterministic} and \emph{stochastic} vertices, $E \subseteq V {\times} V$ is a set of \emph{edges} such that  every vertex has at least one outgoing edge, and $p\colon V_S {\to} \dist(V)$ is a \emph{probability assignment} s.t.{} $p(v)(v') {>} 0$ iff $(v,v') \in E$. A \emph{run} of $D$ is an infinite sequence $\omega = v_1,v_2,\ldots$ such that $(v_i,v_{i+1}) \in E$ for every $i \in \Nset$. The $i$-th vertex $v_i$ visited by $\omega$ is denoted by $\omega(i)$.  We say $D$ is \emph{strongly connected} if the underlying directed graph $(V,E)$ is strongly connected. $D$ is a \emph{graph} if $V_S = \emptyset$.

\paragraph{Strategies}
Outgoing edges in non-deterministic states of an MDP $D = (V,E,p)$ are selected by a \emph{strategy}. The most general type of strategy is a \emph{history-dependent randomized (HR)} strategy where the selection may be randomized and depend on the whole computational history. Since HR strategies require infinite memory, they are not apt for algorithmic purposes.

A strategy is \emph{memoryless randomized (MR)} if the (possibly randomized) decision depends only on the current vertex. Formally, a MR strategy is a pair $\sigma = (v_0,\kappa)$ where $v_0 \in V$ is the \emph{initial} vertex and $\kappa : V \to \dist(V)$ is a function such that $\kappa(v)(u) > 0$ implies $(v,u) \in E$, and for all $v \in V_S$ and $u \in V$ we have that $\kappa(v)(u) = p(v)(u)$. We say that $\sigma$ is \emph{full} if $\kappa(v)(u) > 0$ for all $(v,u) \in E$.

In this paper, we also consider finite-memory randomized strategies with $m\geq 1$ memory states (\emph{FR$_m$ strategies}). Intuitively, the memory states are used to ``remember'' some information about the sequence of previously visited vertices. Formally, let $V' = V \times \{1,\ldots,m\}$ be the set of \emph{augmented vertices}. A FR$_m$ strategy is a pair
$(({v}_0,i_0),\eta)$ where $(v_0,i_0) \in V'$ is an \emph{initial} augmented vertex and $\eta\colon V' \to \dist(V')$ such that $\eta(v,i)(u,j)>0$ implies $(v,u) \in E$. Furthermore, for every $(v,i)$ where $v \in V_S$ and every $(v,u)\in E$ we require $\sum_{j=1}^m \eta(v,i)(u,j) \ = \ p(v)(u)$. Note that every FR$_m$ strategy can be seen as a \emph{memoryless} strategy for an MDP $D'$ where $V'$ is the set of vertices.% and the edges are defined in the expected way.

Let $\xi$ be a strategy (HR, FR$_m$, or MR). For every finite path $v_1,\ldots,v_n$ in $D$, the strategy $\xi$ determines the probability $\prob_\xi(v_1,\ldots,v_n)$ of executing the path. By applying the extension theorem (see, e.g., \cite{Rosenthal:book}), the function $\prob_\xi$ is extended to the probability measure over all runs in $D$.

\paragraph{Strategy profiles}

Let $k \geq 1$. A HR, FR$_m$, MR, or full MR \emph{strategy profile} for $k$ agents is a tuple $\pi = (\xi_1,\ldots,\xi_k)$ where every $\xi_i$ is a HR, FR$_m$, MR, or full MR strategy. A \emph{multi-run} is a tuple $\varrho = (\omega_1,\ldots,\omega_k)$ where each $\omega_i$ is a run of~$D$. We use $\prob_\pi$ to denote the product measure in the product probability space over the set of all multi-runs.

\paragraph{Steady-state objectives}
Let $D = (V,E,p)$ be an MDP and $\Col : V \to \cols$ a \emph{coloring}, where $\cols \neq \emptyset$ is a finite set of colors. A coloring is \emph{trivial} if $\cols {=} V$ and $\Col(v) {=} v$ for every $v \in V$.
%The function $\Col$ is extended to augmented vertices by stipulating $\Col(\ag{u}) = \Col(u)$.

Let  $\pi = (\xi_1,\ldots,\xi_k)$ be a strategy profile and $\varrho = (\omega_1,\ldots,\omega_k)$ a multi-run. For all $c \in \cols$ and $n \geq 1$, we use $\#_c^n(\varrho)$ to denote the total number of all $j \in \{1,\ldots,n\}$ such that  $\Col(\omega_i(j)) = c$ for some $i \in \{1,\ldots,k\}$. Furthermore, we define
\[
    \Freq_c(\varrho) \ = \ \lim_{n \to \infty} \frac{\#_c^n(\varrho)}{n}.
\]
If the above limit does not exist, we put $\Freq_c(\varrho) = {\perp}$.  We use $\Freq(\varrho) : \cols \to [0,1]$ to denote the vector of all $\Freq_c(\varrho)$.

Intuitively, $\Freq_c(\varrho)$ is the long-run average frequency of visits to a $c$-colored vertex by some agent. We say that $\pi$ \emph{achieves} a vector $\mu : \cols \to [0,1]$ if $\prob_\pi[\Freq{=}\mu] = 1$. That is, for every color $c$, the long-run average frequency of visits to a $c$-colored vertex is defined and equal to $\mu(c)$ for almost all multi-runs.

A \emph{steady-state objective} is a vector $\Obj : \cols \to [0,1]$. The task is to construct a strategy profile $\pi$ for $k$ agents such that $\pi$ achieves a vector $\mu \geq \Obj$.

\section{Fundamental Properties of Multi-Agent Steady-State Synthesis}
\label{sec-fundaments}

In this section, we analyze the computational complexity of multi-agent steady-state synthesis. We also investigate the relative power of HR, FR$_m$, MR, and full MR strategy profiles. Proofs of the presented theorems are non-trivial and can be found in \supp.

Let $A$ and $B$ be sets of strategy profiles for an MDP $D$ and $k\geq 1$ agents. Furthermore, let $\Ach(A)$ and $\Ach(B)$ be the sets of all frequency vectors achievable by the profiles of $A$ and $B$, where $\Col$ is the trivial coloring (see Section~\ref{sec-model}). We say that $A$ \emph{approximates}~$B$ if for every $\mu \in \Ach(B)$ and every $\varepsilon > 0$, there is $\nu \in \Ach(A)$ such that $L_\infty(\mu{-}\nu) < \varepsilon$, where $L_\infty(\mu{-}\nu) = \max_{c}(|\mu(c){-}\nu(c)|)$ is the standard $L_\infty$ norm. Furthermore, we say that $B$ is \emph{more powerful} than $A$, written $A \prec B$, if $\Ach(A) \subseteq \Ach(B)$ and $A$~does \emph{not} approximate $B$.

Slightly abusing our notation, we use $\HR(D,k)$, $\FR_m(D,k)$, $\MR(D,k)$, and $\FMR(D,k)$ to denote the sets of all HR, FR$_m$, MR, and full MR strategy profiles for an MDP~$D$ and $k\geq 1$ agents. The next theorem says that the relative power of HR, FR$_m$, MR, and full MR profiles \emph{strictly decreases for $k \geq 2$ agents}, even if $D$ is a strongly connected graph. % (recall that for $k=1$, we have that $\FMR(D,1)$ approximates $\HR(D,1)$ for an arbitrary MDP~$D$; see Section~\ref{sec-intro}).
Since the proof reveals important differences from the single agent scenario, we give a brief sketch.
\begin{theorem}
\label{thm-profile-power}
    There exist strongly connected graphs $D_1, D_2$, and $D_3$ such that
    \begin{itemize}
        \item $\FMR(D_1,2) \prec \MR(D_1,2)$;
        \item $\MR(D_2,2) \prec \FR_2(D_2,2)$;
        \item $\FR_m(D_3,2) \prec \HR(D_3,2)$ for all $m \geq 1$.
    \end{itemize}
\end{theorem}
The graphs $D_1, D_2$, and $D_3$ are shown in Fig.~\ref{fig-profile-power}, together with the frequency vectors achievable by the more powerful strategy profiles that cannot be approximated by the weaker strategy profiles (for $2$ agents). 

In $D_1$, the vector $(1,1)$ is achievable by a MR profile where both agents ``walk around the loop'' connecting $v_1$ and $v_2$, but they start in different vertices. However, for every vector $\nu$ achievable by a FMR profile we have that $\nu(v_2) \leq 0.75$. That is, the $L_\infty$-distance to $(1,1)$ is at least $\delta = 0.25$. Intuitively, this is because the self-loop $v_1 {\to} v_1$ has to be performed with a fixed positive probability, and even if this probability is very small, the two agents spend a \emph{significant} proportion of time by ``walking together'', regardless of their initial positions.
%  (realize that if the two agents walk together, the frequency of visits to $v_1$ and $v_2$ is just $0.5$). 

In $D_2$, a FR$_2$ profile achieving $(1,0.5,0.5)$ consists of strategies where both agents walk around the triangle, performing the self-loop on $u_1$ \emph{exactly once} (this is where two memory states are needed). The first agent starts in $u_1$ by performing the self-loop, and the other agent starts in $u_2$. Thus, the agents never meet, and together they produce the frequency vector $(1,0.5,0.5)$. However, for every vector $\nu$ achievable by a MR profile we have that the $L_\infty$-distance to $(1,0.5,0.5)$ is at least $1/9$. Observe that if both MR strategies assign zero probability to the self-loop on $u_1$, then the frequency of visits to $u_1$ achieved by the profile is at most $2/3$. If at least one of the MR strategies assigns a positive probability to the self-loop, then the two agents spend a significant proportion of time by ``walking together'', similarly as in $D_1$. This leads to the aforementioned gap of $1/9$.
% This leads to mentioned gap~$1/9$. 

The $D_3$ scenario requires deeper analysis. It is easy to show that the vector $(2/3,2/3,2/3,0)$ is achievable by a HR profile where both agents ``walk around the square'' performing each self-loop exactly $n$ times in the $n$-th cycle. Again, the agents are positioned so that they never meet in the same vertex. Furthermore, we show that for every $\nu$ achievable by a FR$_m$ profile, the $L_\infty$ distance to 
$(2/3,2/3,2/3,0)$ is at least $f(m)$ where $f : \Nset_+ {\to} (0,1]$ is a suitable function.

\begin{figure}[t]\centering
    \begin{tikzpicture}[x=1cm, y=1cm, >=stealth',font=\footnotesize]
    \coordinate (d1) at (0,0);
    \node[min] (A1) at (d1) {$v_1$};
    \node[below=.05 of A1] {$1$}; 
    \node[min] (A2) at ($(d1) +(1,0)$) {$v_2$};
    \node[below=.05 of A2] {$1$}; 
    \draw[->,thick,bend left] (A1) to (A2);
    \draw[->,thick,bend left] (A2) to (A1);
    \draw[->,thick,out=150,in=210,looseness=7] (A1) to  (A1);
    \node at ($(d1) + (.5,-1.2)$) {$D_1$};
    \coordinate (d2) at (2.5,0);
    \node[min] (B1) at (d2) {$u_2$};
    \node[below=.05 of B1] {$\frac{1}{2}$}; 
    \node[min] (B2) at ($(d2) +(1,0)$) {$u_3$};
    \node[below=.05 of B2] {$\frac{1}{2}$}; 
    \node[min] (B3) at ($(d2) +(.5,0.8)$) {$u_1$};
    \node[above=.05 of B3] {$1$}; 
    \draw[->,thick] (B3) to (B1);
    \draw[->,thick] (B1) to (B2);
    \draw[->,thick] (B2) to (B3);
    \draw[->,thick,out=150,in=210,looseness=7] (B3) to  (B3);   
    \node at ($(d2) + (.5,-1.2)$) {$D_2$};
    \coordinate (d3) at (5,0);
    \node[min] (C1) at (d3) {$w_1$};
    \node[below=.05 of C1] {$\frac{2}{3}$}; 
    \node[min] (C2) at ($(d3) +(1,0)$) {$w_2$};
    \node[below=.05 of C2] {$\frac{2}{3}$}; 
    \node[min] (C3) at ($(d3) +(1,1)$) {$w_3$};
    \node[above=.05 of C3] {$\frac{2}{3}$}; 
    \node[min] (C4) at ($(d3) +(0,1)$) {$w_4$};
    \node[above=.05 of C4] {$0$}; 
    \draw[->,thick] (C1) to (C2);
    \draw[->,thick] (C2) to (C3);
    \draw[->,thick] (C3) to (C4);
    \draw[->,thick] (C4) to (C1);
    \draw[->,thick,out=30,in=-30,looseness=7] (C2) to (C2);
    \draw[->,thick,out=30,in=-30,looseness=7] (C3) to (C3);
    \draw[->,thick,out=150,in=210,looseness=7] (C1) to (C1);      
    \node at ($(d3) + (.5,-1.2)$) {$D_3$};
\end{tikzpicture}
\caption{The graphs $D_1, D_2$, and $D_3$.}    
\label{fig-profile-power}
\end{figure}
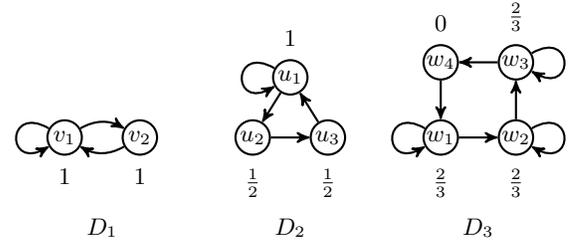

Our next result says that solving the steady-state objectives for $k \geq 2$ agents is computationally hard, even for graphs. 
\begin{theorem}
\label{thm-hardness}
   Let $D$ be a graph, $\Col$ a coloring, and $\Obj$ a frequency vector. We have the following:
   \begin{itemize}
    \item[(a)] The problem whether there exists a HR profile for a given number of agents that achieves $\mu \geq \Obj$ is \PSPACE-hard. This holds even under the assumption that if such a $\mu$ exists, it can be achieved by a FR$_m$ profile for a sufficiently large~$m$.
    \item[(b)] The problem whether there exists a FMR profile for two agents achieving $\mu$ such that $\mu \geq \Obj$ is \NP-hard, even if $D$ is strongly connected and $\Col$ is the trivial coloring. 
    This holds also for MR and FR$_m$ profiles (for every $m$).
 %   The same \NP\ lower bound holds also for MR and FR$_m$ profiles (for every $m \in \Nset$).
   \end{itemize}
\end{theorem}
%Interestingly, the \NP\ lower bound of~(b), does \emph{not} hold for HR profiles. In fact, the existence of a HR profile for two agents achieving $\mu$ such that $\mu \geq \Obj$ is solvable in \emph{polynomial time}. The argument depends crucially on the assumption the $D$ is a graph, and it uses a subtle reduction to the single agent case.

The following theorem reveals a severe obstacle for designing efficient steady-state synthesis algorithms.

\begin{theorem}
\label{thm-evaluate-hard}
    Let $D$ be a (strongly connected) graph, $\Col$ the trivial coloring, $v$ a vertex of $D$, and $\pi$ a MR profile such that $\pi$ achieves some (unknown) frequency vector $\mu$. The problem whether \mbox{$\mu(v)=1$} is \mbox{\coNP-hard}.
\end{theorem}

According to Theorem~\ref{thm-evaluate-hard}, MR strategy profiles are not only hard to construct, but they are also hard to \emph{evaluate}. 
%As we shall see in Section~\ref{sec-algorithm}, this hardness result cannot be extended to full MR profiles which can be evaluated in polynomial time. 

Finally, we give upper complexity bounds on the steady-state synthesis problem.

\begin{theorem}
    Let $k \geq 1$ be a fixed constant. Given an MDP $D$, a coloring $\Col$, a frequency vector $\Obj$, and $m \geq 1$, the problem whether there exists an FR$_m$ strategy profile for $k$ agents achieving $\mu \geq \Obj$ is in \PSPACE\ (assuming the unary encoding of $m$). 
\end{theorem}
\section{Steady-State Synthesis Algorithm}
\label{sec-algorithm}

%In this section, we present an efficient steady-state synthesis algorithm. 
% For a given MDP $D$ in \emph{MEC normal form} (see below), coloring $\Col$, $k \geq 1$, and a frequency vector $\Obj$, the algorithm strives to compute a full MR profile $\pi$ achieving a frequency vector $\mu$ such that $\norm{\mu}_{\Obj}$ is \emph{minimized}, where
% $\norm{\mu}_{\Obj} = \max_{c \in \Colors} \{\max\{0,\Obj(c) - \mu(c)\}\}$.

\paragraph{MDP Normal Form}
We start by observing that in the context of steady-state synthesis, we can safely assume that the input MDP $D$ takes the form $\bigcup_{q=1}^m D_q$ where 
$D_1,\ldots,D_m$ are strongly connected MDPs with pairwise disjoint sets of vertices
(we say that $D$ is in \emph{normal form}).

To see this, consider (some) MDP $D$. A \emph{maximal end component (MEC)} of $D$ is a maximal strongly connected \mbox{sub-MDP} of~$D$. The set $\{D_1,\ldots,D_m\}$ of all MECs of $D$ is computable efficiently \cite{CHH:MEC-decomposition-JACM}, and $D_1,\ldots,D_m$ can be seen as strongly connected MDPs with pairwise disjoint sets of vertices. It can be shown that for an arbitrary (HR) strategy on~$D$, almost all runs eventually enter and stay in some MEC. Since the finite prefix of a run executed before entering the MEC does not influence the achieved frequency vector, we can safely assume that all runs are initiated in some $D_q$ and never leave it. Thus, the steady-state synthesis problem for $D$ can be reformulated as the steady-state synthesis problem for $\bigcup_{q=1}^m D_q$. 
Full details of this argument are somewhat subtle and they are presented in \supp.
% ok - Vojta

Suppose that $\pi$ is a strategy profile for an MDP $\bigcup_{q=1}^m D_q$ in normal form. To compute the frequency vector $\mu$ achieved by $\pi$, one is tempted to compute all frequency vectors $\mu_q$ achieved in $D_q$ by the agents assigned to $D_q$, and then put $\mu = \sum_{q=1}^{m}\mu_q$. However, this simple method works only under the assumption that vertices in different MECs have different colors (we say that $\Col$ is \emph{well-formed}). For example, this condition is satisfied when $\Col$ is the trivial coloring or when $m=1$. If $\Col$ is not well-formed, we can still conclude $\mu \leq  \sum_{q=1}^{m} \mu_q$, but the precise computation of $\mu$ may require \emph{exponential} time, even for full MR profiles. For simplicity, we consider only well-formed colorings in the rest of this section (this condition is not too restrictive and it does not influence the hardness results of Section~\ref{sec-fundaments}).

% Asi nechapu smysl tohoto odstavce, pak jsem se ztratil, ze uz ignoruje non well-formed - Vojta
Let us also note that for MDPs in normal form and one agent, full MR profiles approximate MR profiles, which explains the remark in the second paragraph of Section~\ref{sec-intro}.

% In particular, for $k=1$, we assign the only agent to some $D_i$. If the strategy $\sigma_i$ for $D_i$ is MR, it can be approximated by a full MR strategy $\sigma'_i$ with the same initial vertex up to an arbitrarily small $\delta > 0$. Furthermore, we can extend $\sigma'_i$ to a full MR strategy $\sigma'$ on the whole $\biguplus_{i=1}^m D_i$ (for $D_j \neq D_i$, the functionality of $\sigma'$ is irrelevant and can be chosen arbitrarily). Hence, for $k{=}1$ and MDPs of the form $\biguplus_{i=1}^m D_i$, full MR profiles can approximate MR (and hence also HR) profiles. This explains the claim in the second paragraph of Section~\ref{sec-intro}.

\paragraph{Evaluating Full MR Profiles}

Let $D = (V,E,p)$ be a strongly connected MDP, $\Col : V \to \cols$ a coloring, and $\pi = (\sigma_1,\ldots,\sigma_k)$ a full MR profile for~$D$.  We show how to compute the frequency vector achieved by $\pi$. Note that based on the previous discussion, this procedure can also be used to evaluate a full MR profile for an MDP $\bigcup_{q=1}^m D_q$ in normal form where the underlying coloring is well-formed (we compute the frequency vector $\mu_q$ for each $D_q$ and the agents assigned to $D_q$, and then return the sum of all $\mu_q$).

For every $i \in \{1,\ldots,k\}$, let $D^{\sigma_i} = (V,P_i,\alpha_i)$ be the Markov chain induced by $D$ and $\sigma_i = (v_i,\kappa_i)$. That is, $P_i(v,u)$ is either $\kappa_i(v)(u)$ or $p(v)(u)$ depending on whether $v \in V_N$ or $v\in V_S$, and $\alpha_i(v_i) = 1$. Since $D$ is strongly connected and every $\sigma_i$ is full, each $D^{\sigma_i}$ is irreducible and determines the \emph{same} partition of $V$ into $d \geq 1$ cyclic classes $V_0,\ldots,V_{d-1}$. We use $\Inv_i$ to denote the unique \emph{invariant} distribution of $D^{\sigma_i}$ satisfying $\Inv_i(v) = \sum_{u \in V} \Inv_i(u) \cdot P_i(u,v)$ for every $v \in V$. Furthermore, for every $c \in \cols$, we use $\Col^{-1}(c)$ to denote the pre-image of $c$ (i.e., $\Col^{-1}(c)$ is the set of all $v \in V$ such that $\Col(v)=c$).

For simplicity, let at first consider the case when $d=1$. Then, $\pi$ achieves the frequency vector $\mu$ where
\begin{equation}
    \mu(c) \quad = \quad 1\ - \ \prod_{i=1}^k \bigg(1-\sum_{v\in \Col^{-1}(c)} \Inv_i(v)\bigg)
\label{eq:aperiodic}
\end{equation}
for every $c \in \cols$. This follows directly from basic results about aperiodic irreducible Markov chains (see, e.g., \cite{Chung:book}). More concretely, for every $u \in V$, we have that $\lim_{n \to\infty} P_i^n(v_i,u) = \Inv_i(u)$. Hence, $\sum_{v\in \Col^{-1}(c)} \Inv_i(v)$ is the limit probability that agent~$i$ visits a $c$-colored vertex after $n$~steps as $n {\to} \infty$. Since the agents are independent, the product on the right-hand side of~\eqref{eq:aperiodic} is the limit  probability that \emph{none} of the $k$~agents visits a $c$-colored vertex.
Consequently, the right-hand side of~\eqref{eq:aperiodic} is the limit probability (and hence also the frequency) that \emph{at least one agent} visits a \mbox{$c$-colored} vertex. Note that~\eqref{eq:aperiodic} is independent of the initial vertices of the strategies $\sigma_1,\ldots,\sigma_k$.

If $d>1$, then the frequency vector $\mu$ depends on the initial positioning of the agents into the cyclic classes, and the above reasoning must be applied to the $d$-step matrices $P_i^d$. For every $i \in \{1,\ldots,k\}$ and $j \in \{0,\ldots,d{-}1\}$, let $V(i,j)$ be the cyclic class visited by agent~$i$ after traversing precisely $j$ edges from the initial vertex $v_i$ (in particular, $V(i,0)$ is the cyclic class containing the initial vertex $v_i$). Furthermore, for every $c \in \cols$, let 
$V^c(i,j) = V(i,j) \cap \Col^{-1}(c)$. Equation~\eqref{eq:aperiodic} is generalized into the following:
\begin{equation}
    \mu(c) = \frac{1}{d} \sum_{j=0}^{d-1}\bigg(
       1 - \prod_{i=1}^k \bigg( 1 - d \cdot\sum_{v \in V^c(i,j)} \Inv_i(v) \bigg)\bigg).
\label{eq:periodic}
\end{equation}
Note that~\eqref{eq:periodic} is computable in polynomial time.

\paragraph{The Algorithm}
For a given MDP $D = \bigcup_{q=1}^m D_q$ in normal form, a well-formed coloring $\Col$, $k \geq 1$, and a frequency vector $\Obj$, we wish to compute a full MR profile $\pi$ for $k$ agents achieving a frequency vector $\mu$ such that $\Dist(\mu,\Obj)$ is \emph{minimized}, where 
\begin{equation}
 \Dist(\mu,\Obj) = \sum_{c \in \cols} \max\{0,\Obj(c) - \mu(c)\}.
\end{equation}

A natural idea is to construct a mathematical program minimizing $\Dist(\mu,\Obj)$.
Each full MR strategy $\sigma_i$ in the desired profile $\pi$ can be encoded by variables representing the edge probabilities, and the invariant distribution  $\Inv_i$ can then be encoded by simple linear constraints. However, computing the frequency vector $\mu$ involves the \emph{non-linear} right-hand side of~\eqref{eq:periodic}, which makes the resulting program non-linear.
% Nelinearni uz je to pro aperiodic, to bych zduraznil.

% An \emph{end component} of $D$ is a pair $(V',E')$ where $V' \subseteq V$ and $E' \subseteq E \cap (V' {\times} V')$ such that 
% \begin{itemize}
%     \item for every $v \in V'$, there is an outgoing edge $(v,u) \in E'$;
%     \item if $v \in V_S \cap V'$  and $(v,u) \in E$, then $(v,u) \in E'$;
%     \item $(V',E')$ is strongly connected.
% \end{itemize}
% An end component is \emph{maximal} if it is maximal w.r.t.{} component-wise inclusion.
% Every MDP $D$ with $n$ vertices can be efficiently decomposed into at most $n$ maximal pairwisecol disjoint end-components, and each of these maximal end components can be seen as a strongly connected MDP.

\algnewcommand{\Inputs}[1]{%
  \State \textbf{Inputs:}
  \Statex \hspace*{\algorithmicindent}\parbox[t]{.8\linewidth}{\raggedright #1}
}
\algnewcommand{\Outputs}[1]{%
  \State \textbf{Outputs:}
  \Statex \hspace*{\algorithmicindent}\parbox[t]{.8\linewidth}{\raggedright #1}
}
\algnewcommand{\Initialize}[1]{%
\State \textbf{Initialize:}
\Statex \hspace*{\algorithmicindent}\parbox[t]{.8\linewidth}{\raggedright #1}
}
\begin{algorithm}[t]
	\small
	\caption{Incremental Steady-State Synthesis Algorithm}
	\label{alg:inc-synthesis}
	\begin{algorithmic}
        \Inputs{MDP $D = \bigcup_{q=1}^m D_q$ in normal form\\
                Well-formed coloring $\Col : V \to \cols$\\
                Objective $\Obj : \cols \to [0,1]$\\
                Number of agents $k \geq 1$
        } 
        \Outputs{A full MR strategy profile $\pi$ for $D$ and $k$ agents} 
        \Initialize{$\pi \gets \emptyset$}  
	    \ForAll{$i \in \{1,\ldots, k\}$}
%		    \State $\nu \gets {\rm Evaluate}(\pi,D)$
%           \State ${\rm AuxObj} \gets {\rm ComputeAuxObjective}(\pi,\Obj)$
            \State $\text{BestDistance} \gets \infty$
%            \State $\text{BestStrategy} \gets {\rm some\ full\ MR\ strategy\ for\ } D$
            \ForAll{$q \in \{1,\ldots,m\}$}
                \ForAll{${\rm cyclic\ classes\ } C \in \{C_0,\ldots,C_{d_q-1}\} {\rm\ of\ } D_q$}
                    \State $\sigma \gets \textsc{StrategyOfLP}(\Obj,\pi,C,D_q)$
                    \State $\nu \gets \textsc{Evaluate}(\pi {+} \sigma,D)$
                    \If{${\rm \Dist}(\nu,\Obj) < {\rm BestDistance}}$
                        \State ${\rm BestDistance} \gets {\rm \Dist}(\nu,\Obj)$
                        \State ${\rm BestStrategy} \gets \sigma$
                    \EndIf
                \EndFor
		    \EndFor
            \State $\pi \gets \pi + {\rm BestStrategy}$
		\EndFor
		\State \Return $\pi$
	\end{algorithmic}
\end{algorithm}

To overcome this difficulty, Algorithm~\ref{alg:inc-synthesis} constructs the profile $\pi$ \emph{incrementally} by adding the agents one-by-one. Suppose that we already constructed a profile for $\ell$ agents, and we wish to compute a suitable full MR strategy $\sigma_{\ell+1} = (v_{\ell+1},\kappa_{\ell+1})$ for another agent. The algorithm examines all possible allocations for $v_{\ell+1}$, i.e., all cyclic classes $C$ in all $D_q$. For given $C$ and $D_q$, the procedure 
\textsc{StrategyOfLP} constructs the linear program of Fig.~\ref{fig-lin-prog} and returns the full MR strategy $\sigma = (v_0,\kappa)$, where $v_0 \in C$ and $\kappa(u)(v)$ is the \emph{normalized} value of $x_{u,v}$ attained by solving the program. Note that the $x_{u,v}$ variable in the LP represents the \emph{frequency} of the edge $(u,v) \in E^q$, not the probability of the edge. 
%Hence, for every vertex $v \in V^q$, we have that 
%\(
%\sum_{u \in \Out(v)} x_{v,u} = \sum_{u \in \In(v)} x_{u,v} = \Inv(v).
%\)
The key observation is that since the strategies $\sigma_1,\ldots,\sigma_{\ell}$ are 
\emph{fixed}, the right-hand side of~\eqref{eq:periodic} becomes \emph{linear}. In Fig.~\ref{fig-lin-prog}, we use $\mathcal{X}^c_j$ to denote the \emph{constant} value of the product $\prod_{i=1}^{\ell} (1-d_c\cdot\sum_{v \in V^c(i,j)}\Inv_i(v))$, where $d_c$ denotes the period of the MEC containing the vertices of color $c$ (if there is no such vertex, we put $d_c=1$), $V^c(C,j)$ denotes the set of all $c$-colored vertices in the cyclic class of $D_q$ visited after traversing precisely $j$ edges from a vertex of~$C$. 

After computing the strategy $\sigma$, Algorithm~\ref{alg:inc-synthesis} proceeds by evaluating the profile $\pi+\sigma$ obtained by appending $\sigma$ to $\pi$. If the frequency vector achieved by this profile is better than the frequency vectors achieved for all $\sigma$'s computed so far, the current $\sigma$ is set as a new candidate for $\sigma_{\ell+1}$. 
Algorithm~\ref{alg:inc-synthesis} terminates after constructing a profile for all $k$~agents.

\begin{figure}[t]
\setlength{\tabcolsep}{2pt}\small
\begin{tabular}{l r c l @{\hskip 5.5mm}l}
    \multicolumn{5}{l}{\textbf{min } $\textit{Dist}(\mu,\Obj)$}\\[1ex]
    \multicolumn{5}{l}{\textbf{subject to}}\\[.5ex]    
    & $x_{u,v}$  & $\in$ & $(0,1]$, &  $(u,v)\in E^q$,\\[1ex]
    & $\displaystyle\sum_{(u,v) \in E^q}x_{u,v}$ & $=$ & $1$,\\[4ex]
    & $\displaystyle\sum_{(v,u) \in \Out(v)} x_{v,u}$ & $=$ &  $\displaystyle\sum_{(u,v) \in \In(v)} x_{u,v}$, & $v \in V^q$,\\[3ex]
    \multicolumn{5}{l}{%
     $x_{v,w}$  $=$   $p^q(v)(w) \cdot \displaystyle\sum_{(u,v) \in \In(v)} x_{u,v}$, \hspace*{2em} $v \in V^q_S$, $(v,w) \in \Out(v)$,}\\[3ex]
%    & $\Inv(v)$ & $=$ &  $\displaystyle\sum_{(v,u)\in E^q} x_{v,u}$, & $v \in V^q$,\\
    %& $\displaystyle\sum_{v \in V^q}\Inv(v)$ & $=$ & $1$\\
    % & $\mu(c)$ & $=$ & \multicolumn{2}{l}{%
    %  $\displaystyle\frac{1}{d}\cdot 
    % \sum_{j=0}^{d-1}\bigg(1 - \mathcal{X}^c_j\cdot \bigg(1-d\cdot\bigg(\sum_{v \in V^c(C,j)}\sum_{(u,v) \in \In(v)} x_{u,v}\bigg)\bigg)\bigg)$}\\
     \multicolumn{5}{l}{$\mu(c) = 
    \displaystyle\frac{1}{d_c}\cdot 
   \sum_{j=0}^{d_c-1}\bigg(1 - \mathcal{X}^c_j\cdot \bigg(1-d_c\cdot\sum_{v \in V^c(C,j)}\sum_{(u,v) \in \In(v)} x_{u,v}\bigg)\bigg)$}
\end{tabular}
\caption{The linear program for $\Obj,\pi,C,D_q = (V^q,E^q,p^q)$.
%Recall that $d_q$ is the number of cyclic classes of $D_q$.
}
\label{fig-lin-prog}
\end{figure}
% \begin{tabular}{ l c p{1pt} c }
%     max     & \multicolumn{3}{l}{$300x + 100y$} \\
%     s.t.    & $6x + 3y$ & $\leq$ & $40$\\
%             & $x - 3y$  & $\leq$ & $0$ \\
%             & $x + \frac{1}{4}y$ & $\leq$ & $4$ \\
% \end{tabular}
\section{Experimental Evaluation}
\label{sec-experiments}

The main goal of our experiments is to evaluate the quality of the strategy profiles constructed by Algorithm~\ref{alg:inc-synthesis}. We also assess the efficiency of Algorithm~\ref{alg:inc-synthesis}. 
%In this section, we present only a representative part of our experimental outcomes. 
Additional analyses of the results and some additional plots are
in \supp. The reproduction package for the evaluation is available from Zenodo~\cite{ijcai25_artifact}.

\paragraph{Benchmarks}
For simplicity, we perform our experiments on graphs.
This does not affect efficiency since stochastic vertices do not add any extra computational costs.
Moreover, it does not affect the quality comparison between the baseline and the incremental synthesis procedure of Algorithm~\ref{alg:inc-synthesis}.

To avoid any systematic bias, we randomly generated two
families of strongly connected input graphs: aperiodic and periodic. For
aperiodic graphs, we randomly generated graphs with up to 400 vertices and an
edge between each pair of vertices with probability 0.01. For periodic graphs,
we randomly generated structures of Fig.~\ref{fig-cyclic} with
$d \in \{ 5,10,15,20 \}$ cyclic classes, at most $20$ vertices in each cyclic
class, and an edge between each two vertices from neighboring cyclic classes
with probability $0.6$. We considered only graphs that are strongly connected.
For each graph, we randomly generated $5$ objectives with at most $30$ colors
and target values $\Obj(c)$ from $\{ 0, 0.1, 0.2, \ldots, 0.9 \}$ and randomly
assigned a color to each vertex. In this way, we obtained 2000 aperiodic and
1600 periodic benchmarks (i.e., combinations of a graph and an objective) with
at most 400 vertices.

\paragraph{Baseline}

Since the steady-state synthesis problem for $k {\geq} 2$ agents is computationally hard (see Section~\ref{sec-fundaments}), we cannot compare the quality of profiles constructed by Algorithm~\ref{alg:inc-synthesis} against the optimal solutions as there is no feasible way to determine them. 
However, we can still compare Algorithm~\ref{alg:inc-synthesis} with a \emph{straightforward} synthesis procedure based on sharing the strategy computed for \emph{one} agent.
In some cases, this simple method even leads to optimal solutions. For example, if $k$ agents move along a directed ring consisting of $n \geq k$ vertices, they can achieve the frequency vector $(k/n,\ldots,k/n)$ by sharing the same strategy (``walk along the ring'') so that each agent starts in a different vertex.

The baseline synthesis procedure works as follows. We start by computing a full MR strategy $\sigma$ for a single agent, optimizing a suitably defined value. Then, $k$ agents that use the strategy $\sigma$ are allocated to the cyclic classes $C_0,\ldots,C_{d-1}$ by Round Robin assignment. More specifically, the strategy $\sigma$ is constructed by a LP \emph{maximizing} $\AltDist(\Obj,\nu)$ where $\AltDist(\Obj,\nu) = \min_{c \in \Col}{\frac{\nu(c)}{\Obj(c)}}$ and $\nu(c) = \sum_{v \in \Col^{-1}(c)} \Inv(v)$. Intuitively, the goal is to cover all the colors in proportion to their target values. We maximize $\AltDist$ instead of minimizing $\Dist$ because in our preliminary experiments, the performance of the algorithm based on minimizing $\Dist$ was significantly worse. More concretely, when using $\Dist$, the strategy $\sigma$ tended to focus on a subset of colors, which also caused the resulting strategy profile to focus only on some colors.

\paragraph{Implementation and experimental setup} We implemented both
algorithms in a simple open source Python tool that uses Gurobi~\cite{gurobi} to
solve the linear programming problems. The tool is available from GitLab\footnote{\url{https://gitlab.fi.muni.cz/formela/multi-agent-steady-state-synthesis}}. We executed both algorithms on each
benchmark with timeout 120 seconds of wall time on a Linux computer with
AMD Ryzen 7 PRO 5750G CPU and 32 GB of RAM.

\begin{figure}[tb]
  \centering
  \includegraphics[width=\linewidth]{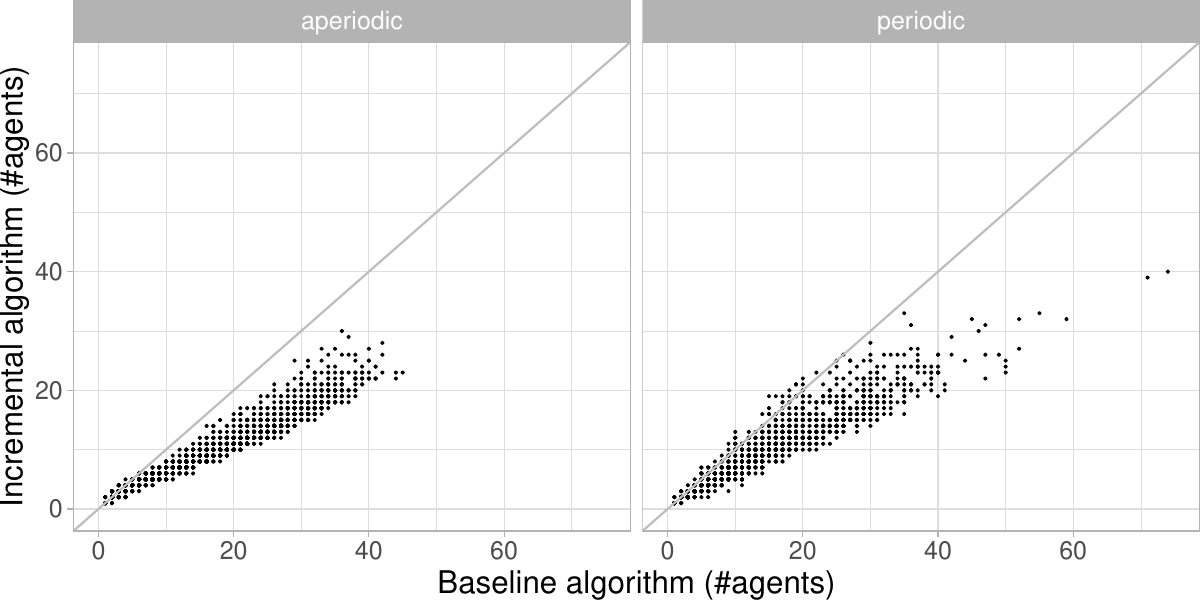}
  \caption{Numbers of agents sufficient to satisfy the objective using each of the algorithms (lower is better). Each point $(x,y)$ is a benchmark for which the objective is satisfied by $x$ agents by the baseline algorithm, and $y$ agents by Algorithm~\ref{alg:inc-synthesis}. Divided by the type of the graph (aperiodic/periodic).}
  \label{fig:agents-scatter}
\end{figure}

\paragraph{Quality of strategies} For each benchmark, we compared the two
algorithms with respect to the number of agents that is necessary to satisfy
the objective. The results are presented in
Fig.~\ref{fig:agents-scatter}. Algorithm~\ref{alg:inc-synthesis} often requires significantly fewer agents to satisfy the objective.
Numerically, Algorithm~\ref{alg:inc-synthesis} required fewer agents on $3163$ of the benchmarks and more on $85$ benchmarks. It also required only $10.40$ agents on
average, compared to $16.65$ agents needed by the baseline. The improvement
occurs both for periodic and aperiodic input graphs, which shows that the main
benefit is not due to the smarter initial assignment of agents to
cyclic classes but because of the core approach of incremental
addition of agents.

We also investigated the distances achieved by the strategy profiles for fewer agents
than necessary to satisfy the objective. This is presented in
Fig.~\ref{fig:distances-during-computation} on a randomly selected subset of
benchmarks. The plot again shows that Algorithm~\ref{alg:inc-synthesis}
satisfies the objective with significantly fewer agents. More importantly, it
shows that for most benchmarks, the profiles synthesized by Algorithm~\ref{alg:inc-synthesis} are better for all numbers of agents smaller than the ones needed by any of the algorithms.

The experiments show that Algorithm~\ref{alg:inc-synthesis} can satisfy objectives
with significantly fewer agents, if enough agents are available. Additionally,
when the number of available agents is insufficient, Algorithm~\ref{alg:inc-synthesis} in most
cases achieves a smaller distance to satisfying the objectives than the
baseline.

\begin{figure}[tb]
  \centering
  \includegraphics[width=\linewidth]{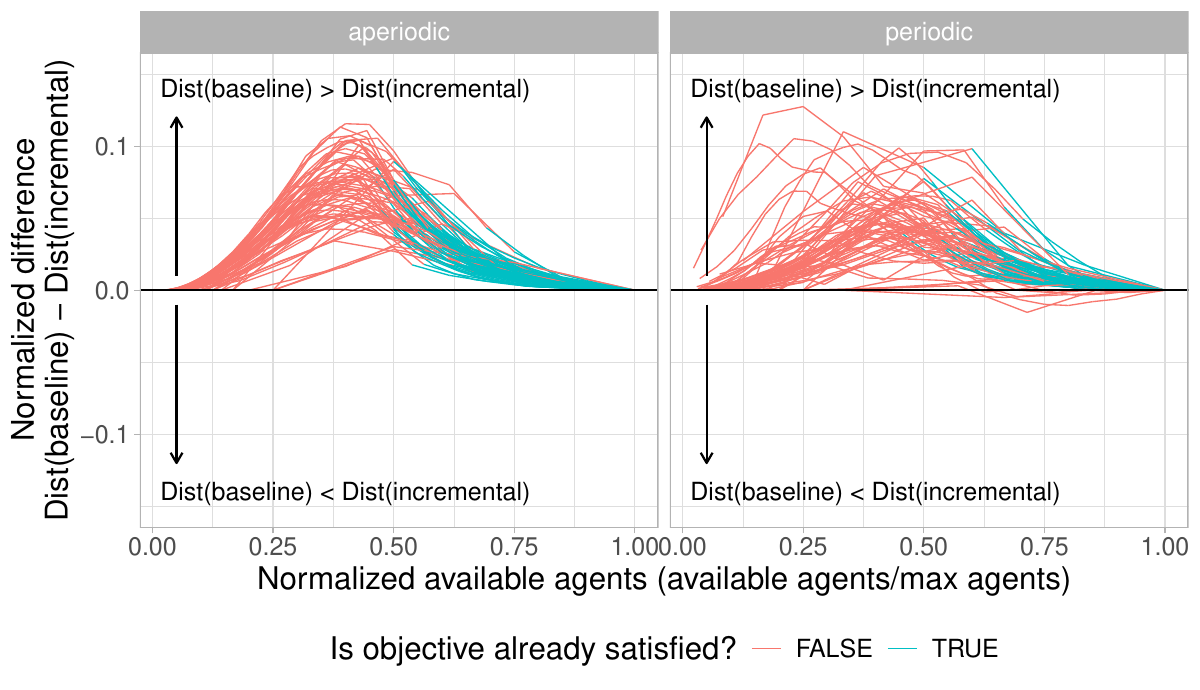}
  \caption{Comparison of distances achieved by the two algorithms on a randomly selected subset of 150 benchmarks. Each line represents a benchmark. The $y$-axis shows the difference of the normalized distances $\frac{\Dist(\pi_{\textrm{baseline}}, \Obj)}{|\cols|} - \frac{\Dist(\pi_{\textrm{incremental}}, \Obj)}{|\cols|}$ between the obtained profiles $\pi_{\textrm{baseline}}$ and $\pi_{\textrm{incremental}}$ for $k$ agents. The $x$-axis shows the number $k$ of agents between $0$ and the number sufficient for both of the algorithms, normalized between $[0,1]$. The line is colored blue if any of the algorithms has already satisfied the objective. Divided by the type of the graph (aperiodic/periodic).}
  \label{fig:distances-during-computation}
\end{figure}

\paragraph{Efficiency}

\begin{table}\small
    \centering
    \begin{tabular}{llrrr}
      \toprule
      & & \multicolumn{3}{c}{Wall time (s)} \\
      \cmidrule{3-5}
      Benchmarks & Algorithm & mean & median & max \\
      \midrule
      \multirow{2}{*}{Aperiodic} & Baseline & 0.10 & 0.09 & 0.25 \\
                                 & Incremental & 0.97 & 0.71 & 5.08\\
      \midrule
      \multirow{2}{*}{Periodic} & Baseline & 0.03 & 0.03 & 0.11 \\
                                & Incremental & 3.57 & 1.88 & 28.87 \\      
      \bottomrule
    \end{tabular}
    \caption{Wall times of both algorithms, divided by the type of the graph
      (aperiodic/periodic).}
    \label{tab:efficiency}
\end{table}

We also measured the runtime for both algorithms. The measured wall times are
summarized in Table~\ref{tab:efficiency}. The table shows that the mean wall
time of the baseline algorithm is significantly better than the one of Algorithm~\ref{alg:inc-synthesis}. This is 
not surprising as the main bottleneck of both algorithms is LP solving
and the baseline algorithm requires only one call
of the LP solver per benchmark, whereas the incremental
algorithm requires $k \cdot d$ calls, where $k$ is the number of agents
and $d$ is the period of the input graph. Nevertheless, all executions of our
algorithm finished within $5.08$ seconds on aperiodic benchmarks and within
$28.88$ seconds on periodic benchmarks, which is not prohibitive in practice.

\paragraph{Discussion} Even though Algorithm~\ref{alg:inc-synthesis} is less
efficient than the naive baseline algorithm, the performance is not prohibitive
in practice and it achieves the objectives with significantly
fewer agents. In our view, this is more important metric; few additional seconds
to synthesize the strategy profile is cheap, whereas each extra agent can be far more
costly.

%%% Local Variables:
%%% mode: LaTeX
%%% TeX-master: "main"
%%% End:

\section*{Conclusions}
We have extended steady-state synthesis to multiagent setting and presented an efficient synthesis algorithm. The main challenges for future work include tackling the synthesis of (non-full) MR profiles and extending the whole approach to more general classes of infinite-horizon objectives.

\newpage
\section*{Acknowledgments}
The work received funding from the European Union’s Horizon Europe program under the Grant Agreement No.\ 101087529.

\bibliographystyle{named}
\bibliography{str-long,concur}

\appendix
\newpage

\section{Proofs for Section~3}
\subsection{A Proof of Theorem 1}
Recall the graphs $D_1$, $D_2$, $D_3$ of Fig.~3 in the main body of the paper.

We start by demonstrating $\FMR(D_1,2) \prec \MR(D_1,2)$. Let $\mu=(1,1)$ and $\varepsilon = \frac{1}{4}$. Note that $\mu$ is achievable by a trivial MR profile for two agents. 
%
% According to the definition, in order to prove that MR profiles are more powerful than FMR profiles for $2$ agents on the graph $D_1$ (written as $\FMR(D_1,2) \prec \MR(D_1,2)$), we have to show that the set of all frequency vectors achievable by FMR profiles is a subset of the set of all frequency vectors achievable by MR profiles (this inclusion is trivial) and that $\FMR(D_1,2)$ does not approximate $\MR(D_1,2)$, meaning that there exists a frequency vector $\mu$ achievable by a MR profile and $\varepsilon > 0$ such that for all frequency vectors $\nu$ achievable by FMR profiles it is true that $L_\infty(\mu,\nu) \geq \varepsilon$.%% $\ell^2(\mu,\nu) \geq \varepsilon$. Since the Euclidean distance $\ell^2(\mu,\nu)$ is always greater than or equal to the Chebyshev distance $L_\infty(\mu,\nu)$, we can afford to prove the statement for $L_\infty(\mu,\nu) \geq \varepsilon$ instead of $\ell^2(\mu,\nu) \geq \varepsilon$.
%
% We prove that the choice $\mu=(1,1)$ and $\varepsilon = \frac{1}{4}$ is sufficient (in the corresponding MR profile achieving $\mu$, the agents are initially placed to distinct vertices and none of the agents uses the self-loop on vertex $v_1$ with nonzero probability). 
%
Let $(\xi_1, \xi_2)$ be an arbitrary FMR profile achieving a frequency vector $\nu$. We  show that $L_\infty(\mu - \nu) \geq \varepsilon$. Observe that for the graph $D_1$, the Markov chains induced by $\xi_1$ and $\xi_2$ are irreducible and aperiodic. Let $\alpha_1$, $\alpha_2$ be the corresponding invariant distributions. %% The missing self-loop on vertex $v_2$ in $D_1$ is enforcing $P_i(v_2,v_2)=0$
Observe that none of the two agents is able to visit $v_2$ more often than in every second step. Since the invariant distribution corresponds to the long-run frequency vector, we have that $\alpha_1(v_2) \leq \frac{1}{2}$, $\alpha_2(v_2)\leq \frac{1}{2}$. The frequency of visits to $v_2$ by some of the two agents can be expressed as 
\[
  1-(1-\alpha_1(v_2))(1-\alpha_2(v_2))\,.
\]  
Thus, we obtain 
\begin{eqnarray*}
  \nu(v_2) & = & 1-(1-\alpha_1(v_2))(1-\alpha_2(v_2))\\
  &\leq& 1-\left(1-\frac{1}{2}\right)\left(1-\frac{1}{2}\right)=\frac{3}{4}
\end{eqnarray*}  
and therefore 
\begin{eqnarray*}
 L_\infty(\mu - \nu) & = & \max\left\{\left|\mu(v_1)-\nu(v_1)\right|, \left|\mu(v_2)-\nu(v_2)\right|\right\}\\
 &\geq& \mu(v_2)-\nu(v_2) \geq 1-\frac{3}{4}=\frac{1}{4} = \varepsilon\,.
\end{eqnarray*}

Now we show $\MR(D_2,2) \prec \FR_2(D_2,2)$. We put $\mu=(1,\frac{1}{2},\frac{1}{2})$ and $\varepsilon = \frac{1}{9}$. A FR$_2$ profile achieving $\mu$ is described already in the main body of the paper (the two agents start in augmented vertices $(u_1,1)$, $(u_2,1)$ and both of them then behave deterministically, transitioning between the augmented vertices in the order $(u_1,1)\mapsto(u_1,2)\mapsto(u_2,1)\mapsto(u_3,1)\mapsto(u_1,1)$).

% In order to show , we are to prove that the set of all frequency vectors achievable by MR profiles is a subset of the set of all frequency vectors achievable by FR$_2$ profiles (this inclusion is again trivial since the agents with FR$_2$ strategies are not forced to ever use the second memory state), it thus remains to show that $\MR(D_2,2)$ does not approximate $\FR_2(D_2,2)$: we conduct the proof similarly as above, this time choosing a frequency vector $\mu=(1,\frac{1}{2},\frac{1}{2})$ and $\varepsilon = \frac{1}{9}$. The FR$_2$ profile achieving $\mu$ is described already in the main text of the article: the two agents start in augmented vertices $(u_1,1)$, $(u_2,1)$ and both of them then behave deterministically, transitioning between the augmented vertices in the order $(u_1,1)\mapsto(u_1,2)\mapsto(u_2,1)\mapsto(u_3,1)\mapsto(u_1,1)$. %% using the memory only in vertex $u_1$ to remember if they used the self-loop on $u_1$ in the previous step: the if the agent has not used the self-loop in the
% It should be easy to see that the described FR$_2$ profile indeed achieves the frequency vector $\mu=(1,\frac{1}{2},\frac{1}{2})$.

Let $(\xi_1, \xi_2)$ be a MR profile achieving a frequency vector~$\nu$. We show that $L_\infty((1,\frac{1}{2},\frac{1}{2}) - \nu) \geq \varepsilon = \frac{1}{9}$.

If none of the two strategies assigns a positive probability to the edge $(u_1,u_1)$, then both agents just walk around the directed triangle $u_1\mapsto u_2\mapsto u_3\mapsto u_1$, and the achieved vector $\nu$ is then either $(\frac{1}{3},\frac{1}{3},\frac{1}{3})$ or $(\frac{2}{3},\frac{2}{3},\frac{2}{3})$, depending on whether the agents start in the same vertex or not. In both cases, $L_\infty((1,\frac{1}{2},\frac{1}{2})-\nu) \geq \frac{1}{3}$. Now assume that at least one of the two strategies assigns positive probability to the edge $(u_1,u_1)$. Observe that the Markov chains induced by $\xi_1$, $\xi_2$ have only one BSCC (it may contain either only $u_1$ or all of the three vertices). Let $\alpha_1$, $\alpha_2$ be the corresponding invariant distributions. Since at least one of the Markov chains is aperiodic, we obtain that 
\begin{eqnarray*}
  \nu(u_i) & = & 1-(1-\alpha_1(u_i))(1-\alpha_2(u_i))\\
  &=& \alpha_1(u_i)+\alpha_2(u_i)-\alpha_1(u_i)\alpha_2(u_i)
\end{eqnarray*}  
for each $i\in\{1,2,3\}$. In the rest of this proof, let $x=\alpha_1(u_1)$ and $y=\alpha_2(u_1)$. %%We assume without loss of generality that $x\leq y$ for the rest of the proof.
Observe that $\frac{1}{3}\leq x\leq 1$, $\frac{1}{3}\leq y\leq 1$, because both agents must visit $u_1$ at least once in every three consecutive steps. Furthermore, the frequency of visits to $u_2$ is equal to the frequency of visits to $u_3$ (this holds for both agents). %%, as it is true for all runs $\omega$ of $D_2$ that every occurrence of $u_2$ in $\omega$ is directly followed by an occurrence of $u_3$ and every occurrence of $u_3$ is directly preceded by an occurrence of $u_2$ (with only one possible exception at the very beginning).
Thus, we get 
\begin{eqnarray*}
 \alpha_1(u_2) & = & \alpha_1(u_3)=\frac{1-x}{2}\,,\\
 \alpha_2(u_2) & = & \alpha_2(u_3)=\frac{1-y}{2}.
\end{eqnarray*}
For the sake of contradiction, assume $L_\infty((1,\frac{1}{2},\frac{1}{2})-\nu) < \frac{1}{9}$. It follows that $\nu(u_1)>1{-}\frac{1}{9}=\frac{8}{9}$ and  $\nu(u_2)>\frac{1}{2}{-}\frac{1}{9}=\frac{7}{18}$. %% and $\nu(u_3)>\frac{1}{2}-\frac{1}{9}=\frac{7}{18}$. %%Using the equalities $\nu(u_i)=\alpha_1(u_i)+\alpha_2(u_i)-\alpha_1(u_i)\alpha_2(u_i)$,
We can thus derive 
\begin{eqnarray*}
  \nu(u_1) &=& \alpha_1(u_1)+\alpha_2(u_1)-\alpha_1(u_1)\alpha_2(u_1)\\
  &=& x+y-xy \ > \ \frac{8}{9}\,,\\
  \nu(u_2) &=& \alpha_1(u_2)+\alpha_2(u_2)-\alpha_1(u_2)\alpha_2(u_2)\\
  &=& \frac{1-x}{2}+\frac{1-y}{2}-\frac{1-x}{2}\frac{1-y}{2}\\
  &>& \frac{7}{18}\,,
\end{eqnarray*} 
where the latter inequality implies $x+y+xy<\frac{13}{9}$.
It follows that 
\begin{eqnarray*}
  xy &=&\frac{1}{2}((x+y+xy)-(x+y-xy))\\
  &<& \frac{1}{2}\left(\frac{13}{9}-\frac{8}{9}\right)=\frac{5}{18}\,.
\end{eqnarray*}
Hence,  $x\frac{1}{3}\leq xy<\frac{5}{18}$ and $x< \frac{5}{6}$. From 
\[
  \frac{8}{9}<x+y-xy=x+y(1-x)
\] and 
\[  
  x+y(1+x)=x+y+xy<\frac{13}{9}\,,
\]
we can derive
{\small
\begin{eqnarray*}
 \frac{8}{9}(1+x) &<& x(1+x)+y(1-x)(1+x)\,,\\
  x(1-x)+y(1-x)(1+x)&<& \frac{13}{9}(1-x)\,,
\end{eqnarray*}}
obtaining
\[
  \frac{8}{9}(1+x)-x(1+x)<y(1-x)(1+x)< \frac{13}{9}(1-x)-x(1-x)\,,
\]
which leads to
\begin{eqnarray*}
 0 & < &\frac{13}{9}(1-x)-x(1-x)-\frac{8}{9}(1+x)+x(1+x)\\
 &=& 2x^2-\frac{7}{3}x+\frac{5}{9} \ =\ 2(x-\frac{1}{3})(x-\frac{5}{6})\,,
\end{eqnarray*}
implying that either $x>\frac{5}{6}$ %%(contradicting with $x<\frac{5}{6}$)
or $x<\frac{1}{3}$. %%%(contradicting with $\frac{1}{3}\leq x$)
In both cases, we obtain a contradiction.%% (with $x<\frac{5}{6}$ or $\frac{1}{3}\leq x$, respectively).

% We are aware of the fact that graph $D_2$ may be generalized to a certain infinite family of graphs $\{D'_1,D'_2,D'_3,\ldots\}$ such that $FR_m(D'_m,2) \prec FR_{m+1}(D'_m,2)$ for all $m\in\mathbb{N}_+$, however, we leave this statement without a proof in this work.

It remains to prove that $\FR_m(D_3,2) \prec \HR(D_3,2)$ for all $m \geq 1$. Let $m \geq 1$ be the number of memory states. Furthermore, let $\mu=(\frac{2}{3},\frac{2}{3},\frac{2}{3},0)$ and $\varepsilon_m = \frac{1}{21m+147}$. 
%Same as above, the set of all frequency vectors achievable by FR$_m$ profiles is trivially a subset of the set of all frequency vectors achievable by HR profiles. It remains to show that $\FR_m(D_3,2)$ does not approximate $\HR(D_3,2)$. We conduct the proof again in an analogous way as before, this time choosing a frequency vector $\mu=(\frac{2}{3},\frac{2}{3},\frac{2}{3},0)$ and $\varepsilon = \frac{1}{21m+147}$.

In this paragraph, we describe two HR strategies such that the HR profile consisting of these two strategies achieves the frequency vector $\mu$.  In the first strategy, the initial vertex is $w_1$. The agent then behaves deterministically, repeating the following iteration for each $n\in\{1,2,3,\ldots\}$ in increasing order: the agent takes $n$ steps, each time using the self-loop on $w_1$, takes a step to $w_2$, takes $n$ steps, each time using the self-loop on $w_2$, takes a step to $w_3$, takes $n$ steps, each time using the self-loop on $w_3$, takes a step to $w_4$, takes a step to $w_1$. In the second strategy, the initial vertex is $w_2$. The agent then behaves deterministically, repeating the following iteration for each $n\in\{1,2,3,\ldots\}$ in increasing order: the agent takes $n$ steps, each time using the self-loop on $w_2$, takes a step to $w_3$, takes $n$ steps, each time using the self-loop on $w_3$, takes a step to $w_4$, takes a step to $w_1$, takes $n$ steps, each time using the self-loop on $w_1$, takes a step to $w_2$.

Since both strategies are deterministic, there is only one possible multi-run determined by the HR profile consisting of these two strategies. Let $n\in\mathbb{N}_+$. Observe that the number of steps performed during the $n$-th iteration is equal to $3n+4$ for each of the two agents, implying that the agents are starting and finishing the iterations simultaneously. During the $n$-th iteration, each of the vertices $w_1$, $w_2$, $w_3$ is visited $2(n+1)$ times (the agents never meet in the same vertex) and the vertex $w_4$ is visited $2$ times. After $k$ complete iterations, $w_4$ has been visited $2k$ times, and each of the remaining three vertices has been visited $\sum_{n=1}^k2(n+1)=k^2+3k$ times, 
%(using the well-known identity $\sum_{n=1}^k n = \frac{k(k+1)}{2}$)
and the total number of steps from the beginning is 
\[
 \sum_{n=1}^k(3n+4)=\frac{3}{2}k(k+1)+4k=\frac{3}{2}k^2+\frac{11}{2}k\,.
\] 
The relative frequency of visits to the vertex $w_4$ after $k$ complete iterations is 
\[
 \frac{2k}{\frac{3}{2}k^2+\frac{11}{2}k}
\] 
and for each of the three remaining vertices, the relative frequency is equal to
\[ 
 \frac{k^2+3k}{\frac{3}{2}k^2+\frac{11}{2}k}\,.
\] 
As $k \to \infty$, the vector of frequencies achieved after $k$ complete iterations approaches $\mu=(\frac{2}{3},\frac{2}{3},\frac{2}{3},0)$. Since the number of steps performed in the course of the $k$-th iteration is $3k+4$, which is only a linear term, it follows that the limit exists even if we consider relative frequencies in all finite prefixes of the multi-run (including prefixes containing unfinished iterations at the end). Hence, the described HR profile achieves $\mu$.

Now let $\pi=(\xi_A, \xi_B)$ be an arbitrary FR$_m$ profile for two agents (referred to as `agent A' and `agent B' in the rest of this proof) achieving the frequency vector $\nu$. We show that 
\[
  L_\infty\left(\left(\frac{2}{3},\frac{2}{3},\frac{2}{3},0\right)-\nu\right) \ \geq \ \frac{1}{21m+147} \ = \ \varepsilon_m\,.
\]  
By definition, almost all multi-runs corresponding to $\pi$ must have the long-run average frequency vector equal to $\nu$. We can assume wlog that each agent starts in an augmented vertex in some BSCC in the Markov chain induced by the corresponding strategy ($\xi_A$ or $\xi_B$). %% and ignore augmented vertices outside this bottom strongly connected component.
Let $\pi_A$ and $\pi_B$ be the invariant distributions of these BSCCs.

% Let the two agents be called $A$ and $B$ and let the stochastic matrix corresponding to the Markov chain induced by $\xi_A$ be denoted by $P$, where the state space $S$ is the set of all augmented vertices, let $\pi_A$ ($\pi_B$, respectively) be the unique invariant distribution corresponding to the bottom strongly connected component where agent $A$ (agent $B$) starts, assigning zero to all augmented vertices outside this bottom strongly connected component.

For the rest of this proof, let $W_i=\{w_i\}\times \{1,\ldots,m\}$ and $\pi_X(w_i)=\sum_{v\in W_i}\pi_X(v)$ for all $i\in\{1,2,3,4\}$ and $X\in\{A,B\}$ (recall that $w_1,\ldots,w_4$ are the vertices of $D_3$). For the sake of contradiction, assume (for the rest of the proof) that 
\[
  L_\infty\left(\left(\frac{2}{3},\frac{2}{3},\frac{2}{3},0\right)-\nu\right) < \varepsilon_m\,,
\] 
implying $\nu(w_j)> \frac{2}{3}-\varepsilon_m$ for each $j\in\{1,2,3\}$. 
Since the sum of all entries of $\nu$ cannot exceed the number of agents, we obtain 
$\nu(w_4)<3\varepsilon_m$.
% The definition allows us to use also the assumption $\nu(w_4)<\varepsilon$, however, that would in some sense considerably weaken the result: in the proofs of the previous parts of this theorem (for graphs $D_1$ and $D_2$), we have actually proved that all frequency vectors achievable by the less powerful strategy profiles are in at least one entry strictly less than $\mu$ by the difference of at least $\epsilon$. In order to repeat that also in this case, we use only a weaker assumption $\nu(w_4)<3\varepsilon$, which follows from the fact that the sum of all entries of $\nu$ cannot exceed the number of agents.%% (which is $2$ in this case).

In this paragraph, we prove that there exists $i\in\{1,2,3\}$ such that %% vertices $w_1$, $w_2$, $w_3$
$\pi_A(w_i)>\frac{1}{7}$ and $\pi_B(w_i)>\frac{1}{7}$. For the sake of contradiction, assume (only for this paragraph) that for some agent $X\in\{A,B\}$ there is at most one $i\in\{1,2,3\}$ such that $\pi_X(w_i)>\frac{1}{7}$. Then, there are $j,k\in\{1,2,3\}$, $j\neq k$ such that $\pi_X(w_j)\leq\frac{1}{7}$, $\pi_X(w_k)\leq\frac{1}{7}$, implying that for the other agent $Y$ we get 
\begin{eqnarray*}
  \pi_Y(w_j) &>& \frac{2}{3}-\varepsilon_m-\pi_X(w_j)\\
  &\geq & \frac{2}{3}-\varepsilon_m-\frac{1}{7}\\
  &=&\frac{11}{21}-\frac{1}{21m+147}\\
  &\geq& \frac{11}{21}-\frac{1}{168}>\frac{1}{2}\,.
\end{eqnarray*}   
Similarly, we obtain $\pi_Y(w_k)>\frac{1}{2}$, which yields a contradiction (it cannot be that $\pi_Y(w_j)+\pi_Y(w_k)>1$, because $\pi_Y$ is a distribution). Thus, there are at least two such $i\in\{1,2,3\}$ for each of the two agents, implying that there is at least one $i\in\{1,2,3\}$ such that $\pi_A(w_i)>\frac{1}{7}$ and $\pi_B(w_i)>\frac{1}{7}$.
% since any two at least two-element subsets of a three-element set have an element in common.

For the rest of the proof, %%let $w=w_i$,
let us fix $i\in\{1,2,3\}$ such that $\pi_A(w_i)>\frac{1}{7}$ and $\pi_B(w_i)>\frac{1}{7}$. Observe that for every $X\in\{A,B\}$, the frequency of $X$'s visits to $w_4$ (that is, $\pi_X(w_4)$) is equal to $X$'s frequency of performing a step along the edge $(w_4,w_1)$, and the same holds also for the edges $(w_1,w_2)$, $(w_2,w_3)$, $(w_3,w_4)$. Since $\pi_X(w_4)\leq\nu(w_4)<3\varepsilon_m$, it follows that the frequency of performing the edge $(w_i,w_{i+1})$ by the agent $X$ must be less than $3\varepsilon_m$. For each $j\in\{1,2,3,4\}$, let $\chi(w_j)$ denote the frequency of simultaneous visits to $w_j$ by both agents $A$ and $B$ (the current assumptions imply that such $\chi(w_j)$ exists and it is unique). The inclusion--exclusion principle implies that $\nu(w_j)=\pi_A(w_j)+\pi_B(w_j)-\chi(w_j)$, and therefore 
\begin{eqnarray*}
  2-3\varepsilon_m & < & \sum_{j=1}^4 \nu(w_j)\\
  & =& \sum_{j=1}^4 (\pi_A(w_j)+\pi_B(w_j)-\chi(w_j))\\
  &=& 1+1-\sum_{j=1}^4 \chi(w_j)\leq 2-\chi(w_i)\,,
\end{eqnarray*}  
from which we obtain $\chi(w_i)<3\varepsilon_m$.%% Our goal in the rest of the proof is to reach a contradiction by showing that $\chi(w_i)$ (the frequency of simultaneous visits of $w_i$ by both agents) must be at least $3\varepsilon$.
%%which implies that the sum of the components of the achieved frequency vector $\nu$ must be at most $2-3\varepsilon$, contradicting the assumption that $\nu(w_j)> \frac{2}{3}-\varepsilon$ for each $j\in\{1,2,3\}$.

In this paragraph, we prove that the Markov chain induced by the strategy of agent $X$ (for all $X \in \{A,B\}$) necessarily contains a directed cycle consisting of transitions with positive probability within the augmented vertices corresponding to the vertex $w_i$ belonging to the BSCC where the agent~$X$ starts. For the sake of contradiction, assume there is no such cycle for agent~$X$. It follows that the agent $X$ can never be present in vertex $w_i$ for more than $m$ consecutive steps, therefore the frequency of performing the edge $(w_i,w_{i+1})$ by agent $X$ is at least 
\[
\pi_X(w_i)\frac{1}{m}>\frac{1}{7m}> \frac{1}{7m+49}=\frac{3}{21m+147}=3\varepsilon_m\,,
\] which contradicts the above statement that this frequency must be less than $3\varepsilon_m$.

Since %%both of the Markov chains
both BSCCs where the agents start contain cycles within the augmented vertices belonging to $w_i$ (there are only $m$ such augmented vertices), it follows that the period of each of these BSCCs is at most $m$. Let $g\leq m$ be the period of the BSCC in the Markov chain induced by $\xi_A$, and let $C_0, C_1, \ldots, C_{g-1}$ be the corresponding cyclic classes.
% (such that for every edge $(v',v'')$ leading between two augmented vertices $v'$, $v''$ in the bottom strongly connected component of the Markov chain there exists $r\in\{0,1,\ldots,g-1\}$ such that $v'\in C_r$ and $v''\in C_{(r+1)\texttt{ mod }g}$). 
%
Since 
\begin{eqnarray*}
 \frac{1}{7} &<& \pi_A(w_i)\\
  &=& \sum_{v\in W_i}\pi_A(v)\\
  &=& \sum_{j=0}^{g-1}\sum_{v\in C_j\cap W_i}\pi_A(v),
\end{eqnarray*}
at least one of these $g$ summands $\sum_{v\in C_j\cap W_i}\pi_A(v)$ must be greater than $\frac{1}{7 g}$. Wlog, we assume it is the summand for $j{=}0$, i.e., 
\[
   \sum_{v\in C_0\cap W_i}\pi_A(v)>\frac{1}{7 g}\,.
\]    
Our next goal is to prove for all $j\in\{0,1,\ldots, g-1\}$ a slightly weaker statement that 
\[
  \sum_{v\in C_j\cap W_i}\pi_A(v)>\frac{1}{7 g}-3\varepsilon_m\,.
\]

Let $k\in\{0,1,\ldots, g-2\}$. We obtain 
\begin{eqnarray*}
 \sum_{v\in C_{k{+}1}\cap W_i}\pi_A(v) &=& \sum_{v\in C_{k+1}\cap W_i}\sum_{u\in S}\pi_A(u)P(u,v)\\
 &\geq&\sum_{v\in C_{k+1}\cap W_i}\sum_{u\in C_k\cap W_i}\pi_A(u)P(u,v)\\
 &=& \sum_{u\in C_k\cap W_i}\pi_A(u)\sum_{v\in C_{k+1}\cap W_i}P(u,v)\\
 &=& \sum_{u\in C_k\cap W_i}\pi_A(u)\big(1-\sum_{v\in C_{k+1}\setminus W_i}P(u,v)\big)\\
 &=& \sum_{v\in C_k\cap W_i}\pi_A(v)\\
 && -\ \sum_{u\in C_k\cap W_i}\sum_{v\in C_{k+1}\setminus W_i}\pi_A(u)P(u,v)\,.
\end{eqnarray*} 

Let $j\in\{0,1,\ldots, g-1\}$. %%By using the inequality $j$ times (for $k=j-1$, $k=j-2$, $\ldots$, $k=0$), we derive that
Using the above inequality, it can be proved by induction on $j$ that 
\begin{eqnarray*}
 \sum_{v\in C_j\cap W_i}\pi_A(v) & \geq & \sum_{v\in C_0\cap W_i}\pi_A(v)\\
 && -\ \sum_{k=0}^{j-1}\sum_{u\in C_k\cap W_i}\sum_{v\in C_{k+1}\setminus W_i}\pi_A(u)P(u,v)\,.
\end{eqnarray*}
Note that the value of the nested sum is at most the frequency of performing the edge $(w_i,w_{i+1})$ by the agent~$A$. Hence, this value must be less than $3\varepsilon_m$. It follows that 
\[
 \sum_{v\in C_j\cap W_i}\pi_A(v)>\frac{1}{7 g}-3\varepsilon_m\,.
\]  
Observe that 
\[
 \frac{1}{7 g}-3\varepsilon_m \geq \frac{1}{7 m}-3\varepsilon_m =\frac{1}{7 m}-\frac{1}{7m+49}>0\,.
\]

Let $h$ be the period of the BSCC containing the augmented vertex where the agent~$B$ starts, and let  $C_0', C_1', \ldots, C_{h-1}'$ be the corresponding cyclic classes. Let $C_a$ and $C_b'$ be the cyclic classes in which the agents $A$ and $B$ start. %%Using the formula expressing $\chi(w_i)$ (the frequency of simultaneous visits of $w_i$ by both agents).
Denoting
\begin{eqnarray*}
   \mathcal{X}_{t} & \equiv & \sum_{v\in C_{(a+t)\texttt{ mod }g}\cap W_i} \pi_A(v)\,,\\[2ex]
   \mathcal{Y}_t & \equiv & \sum_{v\in C_{(b+t)\texttt{ mod }h}'\cap W_i} \pi_B(v)\,,
\end{eqnarray*}

we can now deduce 
\begin{eqnarray*}
 \chi(w_i) &=& \frac{1}{gh}\sum_{t=0}^{gh-1}(g \cdot \mathcal{X}_t)(h \cdot \mathcal{Y}_t)\\
 &\geq& \frac{1}{gh}\sum_{t=0}^{gh-1}(g(\frac{1}{7 g}-3\varepsilon_m))(h \cdot \mathcal{Y}_t)\\
 &=&(\frac{1}{7 g}-3\varepsilon_m)\sum_{t=0}^{gh-1}\sum_{v\in C_{(b+t)\texttt{ mod }h}'\cap W_i} \pi_B(v)\\
 &=&(\frac{1}{7 g}-3\varepsilon_m)g\sum_{v\in W_i} \pi_B(v)\\
 &=&(\frac{1}{7 g}-3\varepsilon_m)g\pi_B(w_i)\\
 &>&(\frac{1}{7 g}-3\varepsilon_m)g\frac{1}{7}\\
 &=&\frac{1}{49}-\frac{3}{7}\varepsilon_m g\\
 &\geq& \frac{1}{49}-\frac{3}{7}\varepsilon_m m\\
 &=&\frac{1}{49}-\frac{3}{7}\frac{1}{21m+147} m\\
 &=&\frac{m+7}{49m+343}-\frac{m}{49m+343}\\
 &=& \frac{3}{21m+147}=3\varepsilon_m\,,
\end{eqnarray*}
 which is in contradiction with previously proved $\chi(w_i)<3\varepsilon_m$. This finishes the proof.

A final remark considering this proof is that the used ``gap estimating function'' $f(m)=\frac{1}{21m+147}$ (also denoted by $\varepsilon_m$ in the proof) cannot be further improved by more than a constant multiplicative factor in the asymptotic sense. The reason is that with $m\geq 2$ available memory states, both agents can deterministically repeat the $(m-1)$-th iteration from the above description of the HR profile achieving the frequency vector $\mu=(\frac{2}{3},\frac{2}{3},\frac{2}{3},0)$ %%(contained in the second paragraph of this part of the proof)
and this deterministic FR$_m$ profile achieves frequency vector $\mu'$, where $\mu'(w_1)=\mu'(w_2)=\mu'(w_3)=\frac{2}{3}-\frac{2}{9m+3}$ and $\mu'(w_4)=\frac{2}{3m+1}$. The $L_\infty$ norm of $\mu-\mu'$ is equal to $\frac{2}{3m+1}$, implying that $f\in\mathcal{O}(\frac{1}{m})$ for any satisfying ``gap estimating function'' $f$. We leave as an open problem whether this asymptotic upper bound is specific to the studied example (graph $D_3$, vector $\mu$) or whether it is a general phenomenon.

\subsection{A Proof of Theorem 2 (a)}
%We prove that the following decision problem is PSPACE-hard. %% TO DO: style of PSPACE
%Let every input instance consist of a graph $D=(U,E, \emptyset)$, %% TO DO: definition in the main text
Let $D = (V,E)$ be a graph, $\Col : V \to \cols$ a coloring, $Obj: \gamma\rightarrow [0,1]$ a frequency vector. We show that the problem whether there exists a HR profile achieving $\mu \geq \Obj$ is \PSPACE-hard.

We prove the result by reduction from the \PSPACE-complete CR-UAV problem for a single UAV \cite{HO:UAV-problem-PSPACE}

An instance of the CR-UAV Problem with a single UAV is a set $V$ of $n\geq 2$ vertices, where each vertex $v\in V$ is assigned a positive integer $RD(v)$ (the ``relative deadline'' of target $v$), and each pair of two distinct vertices $v,v'\in V$ is assigned a positive integer $FT(v,v')$ (the ``flight time'' from $v$ to $v'$). In addition, it is required that $FT$ is symmetric and satisfies the triangle inequality.\footnote{These assumptions are not needed in our proof.} The question is whether there exists an infinite sequence $\zeta = v_0, v_1,\ldots $ (referred to as a \textit{solution} of the problem) such that 
\begin{itemize}
    \item $v_i\neq v_{i+1}$ for all $i\in\Nset$,
    \item every $v \in V$ occurs infinitely often in $\zeta$,
    \item for every finite subsequence $v_i,v_{i+1},\ldots,v_j$ of $\zeta$ that    starts and ends at two consecutive occurrences of $v=v_i=v_j$ we have that $\sum_{t=i}^{j-1}FT(v_t,v_{t+1})\leq RD(v)$. 
\end{itemize} 
Intuitively, the problem is to decide whether a single UAV is able to travel among the vertices of $V$ so that each vertex $v\in V$ is visited infinitely often and the return time is bounded by $RD(v)$ time units. The problem is \PSPACE-complete even if all numerical constants are encoded in unary.

Let $V=\{v_1,\ldots,v_n\}$, $RD$, $FT$ be an instance of the CR-UAV problem with a single UAV with all numerical values of $RD$ and $FT$ represented in unary. In the rest of this proof, we use $N$ to denote the set $\{1,\ldots,n\}$.
We construct a directed graph $D=(U,E)$ where the set of vertices $U$ consists of
\begin{itemize}
    \item all elements of $V$,
    \item vertices of the form $w_{i,j,l}$ where $i,j \in N$, $i \neq j$, $l\in\{1,2,\ldots,FT(v_i,v_j)-1\}$,
    \item vertices of the form $u_{i,j}$ where $i \in N$ and $j\in\{0,\ldots, RD(v_i)-1\}$, 
\end{itemize}
and the set of edges $E$ consists of the following subsets:
\begin{itemize}
    \item $\{(v_i,v_j)\ |\ i,j\in N, FT(v_i,v_j)=1\}$,
    \item $\{(v_i,w_{i,j,1})\ |\ i,j\in N,FT(v_i,v_j)\geq 2\}$,
    \item $\{(w_{i,j,l},w_{i,j,l+1})\ |\ i,j{\in} N,l\in\{1,\ldots, FT(v_i,v_j){-}2\}$,
    \item $\{(w_{i,j,FT(v_i,v_j)-1},v_j)\ |\ i,j\in N,FT(v_i,v_j)\geq 2\}$,
    \item $\{(u_{i,j},u_{i,j-1})\ |\ i\in N, j\in\{1,\ldots, RD(v_i)-1\}\}$,
    \item $\{(u_{i,0},u_{i,RD(v_i)-1})\ |\ i\in N\}$,
    \item $\{(u_{i,j},u_{i,0})\ |\ i\in N, j\in\{0,\ldots, RD(v_i)-1\}$.
\end{itemize}
% Let $N=\{1,2,\ldots, n\}$ for the rest of the proof. In the constructed graph $D=(U,E, \emptyset)$, the set of vertices is set to $U=V\cup\{w_{i,j,l}\ |\ i, j\in N, i\neq j, l\in\{1,2,\ldots,FT(v_i,v_j)-1\}\}\cup\{u_{i,j}\ |\ i\in N, j\in\{0,1,\ldots, RD(v_i)-1\}\}$ and the set of edges to $E=\{(v_i,v_j)\ |\ i,j\in N, FT(v_i,v_j)=1\}\cup\{(v_i,w_{i,j,1})\ |\ i,j\in N,FT(v_i,v_j)\geq 2\}\cup\{(w_{i,j,l},w_{i,j,l+1})\ |\ i,j\in N,l\in\{1,2,\ldots, FT(v_i,v_j)-2\}\cup\{(w_{i,j,FT(v_i,v_j)-1},v_j)\ |\ i,j\in N,FT(v_i,v_j)\geq 2\}\cup\{(u_{i,j},u_{i,j-1})\ |\ i\in N, j\in\{1,2,\ldots, RD(v_i)-1\}\}\cup\{(u_{i,0},u_{i,RD(v_i)-1})\ |\ i\in N\}\cup\{(u_{i,j},u_{i,0})\ |\ i\in N, j\in\{0,1,\ldots, RD(v_i)-1\}$. 

The set of colors is 
\[
  \cols=\{c_i\ |\ i\in N\}\cup\{c'_i\ |\ i\in N\}\cup\{\hat{c}\}\,.
\]
These colors are assigned to the vertices of $U$ by coloring $\Col: U\rightarrow \gamma$, where 
\begin{itemize}
    \item $\Col(v_i)=c_i$ for all $i\in N$,
    \item $\Col(u_{i,0})=c'_i$ for all $i\in N$, 
    \item $\Col(u_{i,j})=c_i$ for all $i\in N$, $j\in\{1,\ldots,RD(v_i)-1\}$,
    \item $\Col(w_{i,j,l})=\hat{c}$ for all $w_{i,j,l}$.
\end{itemize}
The objective $\Obj: \cols\rightarrow [0,1]$ is defined as follows: 
\begin{itemize}
    \item $\Obj(\hat{c})=0$, 
    \item $\Obj(c_i)=1$, 
    \item $\Obj(c'_i)=\frac{1}{RD(v_i)}$ for all $i\in N$.
\end{itemize}
The number of agents is $k=n+1$.

Intuitively, each pair of distinct vertices $v_i,v_j\in V$ is connected (in both directions) by a directed path of length $FT(v_i,v_j)$ leading through the newly added vertices $w_{i,j,l}$. Each vertex $v_i$ has its associated ``timer gadget'' consisting of a directed cycle 
\[u_{i,0}\mapsto u_{i,RD(v_i)-1}\mapsto u_{i,RD(v_i)-2}\mapsto \ldots \mapsto u_{i,1}\mapsto u_{i,0}
\]
and additional edges leading to $u_{i,0}$ from all vertices of the gadget. The constructed graph has $n+1$ strongly connected components, matching the number of agents $k$ (each gadget forms one component, the remaining part is strongly connected). Each target vertex $v_i$ is assigned the same color as the vertices in its associated timer gadget, except for the vertex $u_{i,0}$, which has its own color. When an agent is put into a timer gadget associated to a target vertex $v_i$, it can at any moment enter the vertex $u_{i,0}$ (meaning to ``stop the timer''). If the agent in the timer gadget is in $u_{i,0}$, it can take a step into $u_{i,RD(v_i)-1}$ (meaning to ``reset the timer''). If the agent in the timer gadget is in $u_{i,j}\neq u_{i,0}$, it can take a step into $u_{i,j-1}$ (meaning to ``continue the countdown''). Observe that the agent in the timer gadget is forced to step into $u_{i,0}$ at least once in every $RD(v_i)$ consecutive steps and whenever the agent is in $u_{i,0}$, it is capable of not returning to $u_{i,0}$ earlier than after $RD(v_i)$ steps. The construction is designed in a way that if there exists a HR profile for $k=n+1$ agents that achieves the objective $\Obj$, then one agent has to be put into each of the timer gadgets and the remaining agent has to be put into the remaining strongly connected component: $n$ agents thus play the role of ``timers'' and the remaining one agent plays the role of a UAV.

We prove that $V$, $RD$, $FT$ is a positive instance of the CR-UAV Problem with a single UAV if and only if there exists a HR profile for $k=n+1$ agents achieving some frequency vector $\mu$ such that $\mu\geq Obj$.

Assume that $V$, $RD$, $FT$ is a positive instance of the CR-UAV Problem with a single UAV. As mentioned in \cite{HO:UAV-problem-PSPACE}, %%the proof contained in \cite{HO:UAV-problem-PSPACE} implies that
if there is a solution to the problem (a sequence $s\in V^{\omega}$ satisfying the requirements), then there exists also a periodic solution in which the target vertices are visited repeatedly in the same order, meaning that there exist $r\in \mathbb{N}_+$ and $\tilde{s}=(\tilde{s}_0, \tilde{s}_1, \ldots, \tilde{s}_{r-1})\in V^{r}$ having the following properties: $\tilde{s}_{r-1}\neq\tilde{s}_0$, for all $i\in\{0,1,\ldots,r-2\}$ it holds that $\tilde{s}_i\neq\tilde{s}_{i+1}$, for the infinite sequence $s'=\tilde{s}^\omega\in V^{\omega}$ it holds that each $v\in V$ has infinitely many occurrences in $s'$ and for any finite substring $(s'_i,s'_{i+1},\ldots,s'_j)$ of $s'$ which starts and ends at two consecutive occurrences of $v=s'_i=s'_j$ we have that $\sum_{t=i}^{j-1}FT(s'_t,s'_{t+1})\leq RD(v)$. Let such $r$, $\tilde{s}$ and the corresponding $s'=\tilde{s}^\omega$ be fixed throughout the rest of this proof.

Let us define the following multi-run $\rho$ for $k=n+1$ agents $A_0,A_1,\ldots,A_n$ and the constructed graph $D=(U,E)$: initially, the agent $A_0$ is placed to $\tilde{s}_0$ and agent $A_i$ is placed to $u_{i,0}$ for each $i\in N$; the agent $A_0$ then visits vertices of $V$ in the same order as in $s'$, using the vertices $w_{i,j,l}$ to travel between them; whenever the agent $A_0$ is to take a step into $v_i$ (for any $v_i\in V$), $A_i$ takes a step into $u_{i,0}$ (so that every time $A_0$ is in $v_i$, $A_i$ is in $u_{i,0}$), otherwise,  $A_i$ uses the only edge that does not lead to $u_{i,0}$, provided that such edge exists (if not, then $A_i$ steps into $u_{i,0}$). %% (since the input instance of the CR-UAV Problem with a single UAV is positive, it follows that $RD(v_i)\geq 2$).
For this multi-run $\rho$ it holds that after $A_0$ has visited each of the vertices in $V$ at least once, then for all $i\in N$ and all $j\in\{0,1,\ldots,RD(v_i)-1\}$ we have that $A_i$ is in $u_{i,j}$ if and only if the last visit of vertex $v_i$ by $A_0$ has occurred exactly before $(-j)\texttt{ mod }RD(v_i)$ steps (if $A_0$ is present in $v_i$, it counts as the last visit, occurring exactly before $0$ steps).

Let multi-run $\rho'=(\rho'_0, \rho'_1, \ldots, \rho'_n)$ be a suffix of the multi-run $\rho$ that is obtained from $\rho$ by dropping the first $m=FT(\tilde{s}_{r-1}, \tilde{s}_0)+\sum_{t=0}^{r-2}FT(\tilde{s}_t, \tilde{s}_{t+1})$ elements (from each entry of $\rho$): the initial configuration of agents in $\rho'$ corresponds to the configuration that is reached in $\rho$ after $A_0$ has visited all vertices in $\tilde{s}$ and returned to $\tilde{s}_0$. Using the previous observations, we obtain that the configurations of agents in the multi-run $\rho'$ repeat with a constant period $m$. Let $\pi$ be the HR profile fixing $\rho'$, meaning that the strategy of each of the $n+1$ agents $A_i$ deterministically emulates $\rho'_i$. We claim that $\pi$ achieves some frequency vector $\mu\geq Obj$. The fact that there indeed exists a frequency vector achieved by $\pi$ follows from the fact that the configurations of agents in the multi-run $\rho'$ (the only possible multi-run corresponding to the strategy profile $\pi$) periodically repeat with a finite period of $m$ steps: the long-run average frequency of visits to vertices of every color $c$ is thus equal to the average frequency of such visits within the first $m$ steps (which is always a rational number).

Let $\mu$ be the frequency vector achieved by $\pi$. It holds trivially that $\mu(\hat{c})\geq Obj(\hat{c})=0$. The multi-run $\rho'$ is defined in a way that for all $i\in N$ there is always an agent in some vertex of the color $c_i$: when agent $A_0$ (playing the role of a UAV) is in $v_i$ (of color $c_i$), $A_i$ is in $u_{i,0}$ and since $A_0$ always returns to $v_i$ after at most $RD(v_i)$ steps, it means that agent $A_i$ then stays in vertices of color $c_i$ until $A_0$ returns to $v_i$ (as already mentioned, $A_i$ is capable of not returning to $u_{i,0}$ earlier than after $RD(v_i)$ steps and by definition of $\rho'$, if $A_0$ does not step into $v_i$, $A_i$ does not enter $u_{i,0}$ unless it is forced to do so). Therefore $\mu(c_i)\geq Obj(c_i)=1$ for all $i\in N$. It follows from the construction that vertices $u_{i,0}$ are visited with frequency at least $\frac{1}{RD(v_i)}$, so $\mu(c'_i)\geq Obj(c'_i)=\frac{1}{RD(v_i)}$ for all $i\in N$. We have thus shown that $\mu\geq Obj$, finishing the proof of the implication. Notice that each strategy in the profile $\pi$ tells the agent to deterministically repeat a sequence of $m$ steps (ad infinitum), which is possible to implement with only $m$ memory states. It follows that there is actually a FR$_m$ profile satisfying the objective.
%%Let $s=(s_0, s_1, s_2,\ldots)\in V^{\omega}$ be an infinite sequence such that for all $v\in V$ it holds that $v$ has infinitely many occurrences in $s$ and at the same time for any finite substring $(s_i,s_{i+1},\ldots,s_j)$ of $s$ which starts and ends at two consecutive occurrences of $v=s_i=s_j$ we have that $\sum_{t=i}^{j-1}FT(s_t,s_{t+1})\leq RD(v)$.%% Since $V$ is finite, there is an index $h\in\mathbb{N}$ such that each $v\in V$ occurs at least once in the finite prefix $(s_0, s_1, \ldots, s_h)$. It follows that the remaining suffix $s'=(s_{h+1}, s_{h+2}, s_{h+3},\ldots)$ has the property that for every $v_i\in V$ there is at least one occurrence of $v_i$ within any $RD(v_i)$ consecutive

In order to prove the converse, assume there exists a HR profile for $k=n+1$ agents (and the constructed graph $D=(U,E)$) achieving some frequency vector $\mu$ such that $\mu\geq Obj$. We are to prove that $V$, $RD$, $FT$ is a positive instance of the CR-UAV Problem with a single UAV. Let $\tau$ be a multi-run of $D$ with the corresponding frequency vector $\mu\geq Obj$: since there exists a HR profile satisfying the objective with probability $1$ (according to the definition), it follows that such $\tau$ exists. Observe that in the construction of $D$, $Col$ and $Obj$, only vertex $u_{i,0}$ is assigned color $c'_i$, thus at least one agent has to start in the strongly connected component containing $u_{i,0}$ (for all $i\in N$) in order to satisfy the nonzero objective for vertex $u_{i,0}$. It follows that the remaining agent has to start in the remaining strongly connected component of $D$, otherwise, there is a color $c_i$ for which the objective is not satisfied: remember that $n\geq 2$ and that none single agent %%in the strongly connected component containing $u_{i,0}$
is able to visit vertices of color $c_i$ with long-run average frequency equal to $1$ (for all $i\in N$).

Observe that the number of all possible configurations of $n+1$ agents in a graph with $|U|$ vertices is bounded from above by $|U|^{n+1}$. In the multi-run $\tau$, vertices of each of the colors $c_1, c_2, \ldots, c_n$ are being visited by at least one agent (being ``occupied'') with long-run average frequency equal to $1$ (almost always), which implies that the long-run average frequency of simultaneous occupation of all colors $c_1, c_2, \ldots, c_n$ (that is, when for each $i\in N$ there is at least one agent visiting a vertex of color $c_i$) is also equal to $1$. 
% we omit the proof of this argument, considering its validity to be rather obvious (it follows from the fact that there is only a finite number of colors involved). 
Thus, in the course of the multi-run $\tau$ there eventually occur $|U|^{n+1}+1$ consecutive steps during which all of the colors $c_1, c_2, \ldots, c_n$ remain occupied. %% (meaning that there is at least one agent in some of the vertices of color $c_i$ in each of these consecutive steps for all $i\in N$)

Therefore, there is a configuration of agents from which the same configuration of agents can be reached in a positive number of (at most $|U|^{n+1}$) steps so that vertices of any of the colors $c_1, c_2, \ldots, c_n$ remain occupied. Consider the order in which the vertices of $V$ are visited along such a way from such a configuration to itself (notice that all vertices of $V$ indeed have to be visited and that all these visits are performed by the same agent): it follows from the construction that repeating this sequence of vertices of $V$ ad infinitum provides a (periodic) solution to the input instance $V$, $RD$, $FT$ of the CR-UAV Problem with a single UAV.
%% In the multi-run $\tau$, each of the colors $c_1, c_2, \ldots, c_n$ is being visited by at least one agent with long-run average frequency equal to $1$ (almost always), which implies that there eventually occur $|U|^{n+1}+1$ consecutive steps during which all of the colors $c_1, c_2, \ldots, c_n$ remain occupied (meaning that there is at least one agent in some of the vertices of color $c_i$ in each of these consecutive steps for all $i\in N$): we omit the proof of this argument, considering its validity to be rather obvious (it follows from the fact that there is only a finite number of colors involved).
\subsection{A Proof of Theorem 2 (b)}

We prove these hardness results by reduction from SAT (the Boolean satisfiability problem), which is NP-complete. 
%The same reduction is used to prove the result for the case of MR strategy profiles and the case of FMR strategy profiles. 
Let $\psi$ be a propositional formula in conjunctive normal form. Without restrictions, we assume that $\psi$ contains at least two clauses, at least two distinct propositional variables, and no tautological clauses (that is, clauses containing both $p$ and $\neg p$ for some propositional variable $p$). Furthermore, we assume that if a propositional variable $q$ occurs in $\psi$, then both of the literals $q$ and $\neg q$ occur in $\psi$.  
% Let $\psi$ be a propositional formula in conjunctive normal form equisatisfiable with $\varphi$ where $\psi$ contains at least two clauses, at least two distinct propositional variables, no tautological clauses (that is, clauses containing both $p$ and $\neg p$ for some propositional variable $p$) %% TO DO: is this assumption necessary?
% and where for every propositional variable $q$ occurring in $\psi$ there is a clause of $\psi$ containing the corresponding negative literal $\neg q$ and a clause containing the corresponding positive literal $q$,

% Finally, the set of all propositional variables occurring in $\psi$ is required to be equal to $\{q_1, q_2, \ldots, q_n\}$ for some $n\geq 2$: %% TO DO: definition of literals (should be already present elsewhere in this file)
% such $\psi$ can be obtained in time polynomial in the length of $\varphi$ by removing all tautological clauses and then subsequently removing (repeatedly) clauses containing a variable for which only its corresponding negative or only its corresponding positive literal occurs in (the current version of) the formula, clauses $(q'\lor\neg \hat{q})$ and $(\neg q'\lor \hat{q})$ (where $q'$ and $\hat{q}$ are fresh variables) can be appended to ensure that the formula contain at least two clauses and at least two variables, finally, all occurring variables may be renamed to achieve the required form of $\psi$.

Let $\psi_1,\psi_2,\ldots,\psi_r$ be the clauses of $\psi$, and let  $q_1, q_2, \ldots, q_n$ 
be the propositional variables occurring in $\psi$.
%% In the rest of the description...
For the rest of this proof, we fix the following sets:
$U'=\{u_1,\ldots,u_n\}$, $V'=\{v_1,\ldots,v_n\}$, $W'=\{w_1,\ldots,w_n\}$.
% throughout the rest of this proof. %% (where $|U'|=|V'|=|W'|=n$, $U'\cap V'=U'\cap W'=V'\cap W'=\emptyset$). 

Let $D=(V,E)$ be a graph where
\[
  V=\{u,v,w\}\cup U'\cup V'\cup W'\cup\{x_1,x_2,\ldots,x_r\}
\]  
and 
\begin{eqnarray*}
  E & = &  \{(u,v),(v,v),(w,w)\}\\
    & \cup & \{(v,u_1),(v,u_2),\ldots,(v,u_n)\}\\
    & \cup & \{(u_1,v_1),\ldots,(u_n,v_n)\}\\
    & \cup & \{(u_1,w_1),\ldots,(u_n,w_n)\}\\
    & \cup & \{(w,v_1),\ldots,(w,v_n)\}\\
    & \cup & \{(w,w_1),\ldots,(w,w_n)\}\\
    & \cup & \{(v_1,w),\ldots,(v_n,w)\}\\
    & \cup & \{(w_1,w),\ldots,(w_n,w)\}\\
    & \cup & \{(x_1,u),\ldots,(x_r,u)\}\\
    & \cup & \{(w_i,x_j) | \neg q_i \in \psi_j\}\\
    & \cup & \{(v_i,x_j) | q_i \in \psi_j\}\,.
\end{eqnarray*}

% the set of directed edges being $E=\{(u,v),(v,v),(w,w)\} \cup \{(v,u_1),(v,u_2),\ldots,(v,u_n)\} \cup \{(u_1,v_1),(u_2,v_2),\ldots,(u_n,v_n)\} \cup \{(u_1,w_1),(u_2,w_2),\ldots,(u_n,w_n)\} \cup \{(w,v_1),(w,v_2),\ldots,(w,v_n)\} \cup \{(w,w_1),(w,w_2),\ldots,(w,w_n)\} \cup \{(v_1,w),(v_2,w),\ldots,(v_n,w)\} \cup \{(w_1,w),(w_2,w),\ldots,(w_n,w)\} \cup \{(x_1,u),(x_2,u),\ldots,(x_r,u)\}\cup \{(w_i,x_j) | \neg q_i \in \psi_j\} \cup \{(v_i,x_j) | q_i \in \psi_j\}$, where $\neg q_i \in \psi_j$ means that clause $\psi_j$ contains the negative literal $\neg q_i$ and $q_i \in \psi_j$ means that $\psi_j$ contains the positive literal $q_i$ (note that $\psi_j$ always contains at least one literal: empty clauses are not allowed to occur in $\psi$).

Let $\Col$ be the trivial coloring (that is, $\Col(z)=z$ for all $z\in V$). Let \[\zeta=\frac{1}{1024 n^4 r^2}\] throughout the rest of the proof. The frequency vector $\Obj$ is defined as follows: 
\begin{itemize}
    \item $\Obj(u)=\frac{1-\zeta}{8}$,
    \item $\Obj(v)=\frac{1-\zeta}{2}$, 
    \item $\Obj(w)=\frac{7(1-\zeta)}{8}$, 
    \item $\Obj(u_i)=\Obj(v_i)=\Obj(w_i)=\frac{1-\zeta}{8n}$ for all $i{\in}\{1,\ldots,n\}$,
    \item  $\Obj(x_j)=\frac{1-\zeta}{8nr}$ for all $j\in\{1,\ldots,r\}$. 
\end{itemize}
% $Obj(u)=\frac{1-\zeta}{8}$, $Obj(v)=\frac{1-\zeta}{2}$, $Obj(w)=\frac{7(1-\zeta)}{8}$, $Obj(u_i)=Obj(v_i)=Obj(w_i)=\frac{1-\zeta}{8n}$ for all $i\in\{1,2,\ldots,n\}$, $Obj(x_j)=\frac{1-\zeta}{8nr}$ for all $j\in\{1,2,\ldots,r\}$. 
This completes the description of the reduction. 
% (it should be rather obvious that the construction of $D$, $Col$ and $Obj$ can be implemented by a polynomial-time algorithm). 
We prove that the formula $\psi$ is satisfiable if and only if there exists a MR strategy profile (or FMR strategy profile) $\pi=(\xi_A,\xi_B)$ for $k=2$ agents achieving a frequency vector $\mu\geq \Obj$.

%In order to prove the `$\Rightarrow$' direction of the logical equivalence, 
Assume that $\psi$ is satisfiable. We prove that there exists a MR profile (or FMR profile) $\pi=(\xi_A,\xi_B)$ for $k=2$ agents achieving a frequency vector $\mu\geq Obj$. We start by describing the MR profile, and then show how to modify this profile into a FMR profile. % at the end of the proof of the `$\Rightarrow$' direction. Since $\varphi$ is equisatisfiable with $\psi$, we have that $\psi$ is also satisfiable. 
Let $\vartheta$ be a valuation 
%(that is, an assignment of truth values $0$ or $1$ to all propositional variables, being generalized to an assignment of truth values to all propositional formulae) 
such that $\vartheta(\psi)=1$. Let $g$ be a function assigning to each literal the number of clauses of $\psi$ containing the literal. In the constructed MR strategy profile $\pi=(\xi_A,\xi_B)=((v_A,\kappa_A),(v_B,\kappa_B))$, we put $v_A = v$ and $v_B = w$. 
%let the initial vertex $v_A$ of agent $A$ be equal to $v$ and the initial vertex $v_B$ of agent $B$ equal to $w$. 
%Let $\kappa_A(x)(y)=\kappa_B(x)(y)=0$ for all vertices $x,y\in V$ such that $(x,y)\notin E$ (as the definition requires). 
The functions $\kappa_A$ and $\kappa_B$ are defined as follows. For all 
$i\in\{1,\ldots,n\}$ and $j\in\{1,\ldots,r\}$, we put
\begin{itemize}
    \item $\kappa_A(x_j)(u)=\kappa_B(x_j)(u)=1$, 
    \item $\kappa_A(u)(v)=\kappa_B(u)(v)=1$, 
    \item $\kappa_A(v)(v)=\frac{3}{4}$, 
    \item $\kappa_B(v)(v)=0$, 
    \item $\kappa_A(v)(u_i)=\frac{1}{4n}$, 
    \item $\kappa_B(v)(u_i)=\frac{1}{n}$, 
    \item $\kappa_A(u_i)(v_i)=\vartheta(q_i)$, 
    \item $\kappa_A(u_i)(w_i)=\vartheta(\neg q_i)$, 
    \item $\kappa_B(u_i)(v_i)=\kappa_B(u_i)(w_i)=\frac{1}{2}$, 
    \item $\kappa_A(v_i)(w)=\kappa_A(w_i)(w)=0$, 
    \item $\kappa_B(v_i)(w)=\kappa_B(w_i)(w)=1$, 
    \item $\kappa_A(w)(v_i)=\kappa_A(w)(w_i)=\frac{1}{2n}$, 
    \item $\kappa_B(w)(v_i)=\frac{1}{7n}\vartheta(\neg q_i)$, 
    \item $\kappa_B(w)(w_i)=\frac{1}{7n}\vartheta(q_i)$, 
    \item $\kappa_A(w)(w)=0$, 
    \item $\kappa_B(w)(w)=\frac{6}{7}$, 
    \item $\kappa_A(v_i)(x_j)=\frac{1}{g(q_i)}$ (unless $(v_i, x_j)\notin E$),
    \item $\kappa_A(w_i)(x_j)=\frac{1}{g(\neg q_i)}$ (unless $(w_i, x_j)\notin E$),
    \item $\kappa_B(v_i)(x_j)=\kappa_B(w_i)(x_j)=0$.
\end{itemize}
% $\kappa_A(x_j)(u)=\kappa_B(x_j)(u)=1$, $\kappa_A(u)(v)=\kappa_B(u)(v)=1$, $\kappa_A(v)(v)=\frac{3}{4}$, $\kappa_B(v)(v)=0$, $\kappa_A(v)(u_i)=\frac{1}{4n}$, $\kappa_B(v)(u_i)=\frac{1}{n}$, $\kappa_A(u_i)(v_i)=\vartheta(q_i)$, $\kappa_A(u_i)(w_i)=\vartheta(\neg q_i)$, $\kappa_B(u_i)(v_i)=\kappa_B(u_i)(w_i)=\frac{1}{2}$, $\kappa_A(v_i)(w)=\kappa_A(w_i)(w)=0$, $\kappa_B(v_i)(w)=\kappa_B(w_i)(w)=1$, $\kappa_A(w)(v_i)=\kappa_A(w)(w_i)=\frac{1}{2n}$, $\kappa_B(w)(v_i)=\frac{1}{7n}\vartheta(\neg q_i)$, $\kappa_B(w)(w_i)=\frac{1}{7n}\vartheta(q_i)$, $\kappa_A(w)(w)=0$, $\kappa_B(w)(w)=\frac{6}{7}$, $\kappa_A(v_i)(x_j)=\frac{1}{g(q_i)}$ (unless $(v_i, x_j)\notin E$), $\kappa_A(w_i)(x_j)=\frac{1}{g(\neg q_i)}$ (unless $(w_i, x_j)\notin E$), $\kappa_B(v_i)(x_j)=\kappa_B(w_i)(x_j)=0$.

The Markov chains induced by $\xi_A$ and $\xi_B$ contain a single BSCC, and this BSCC is aperiodic (because of the presence of a self-loop). Furthermore, these two BSCCs are disjoint (the agents can never meet in the same vertex). Let $\alpha$, $\beta$ be the unique invariant distributions corresponding to the two induced Markov chains. It can be easily shown that $\beta(w)=\frac{7}{8}$ and that for each $i\in\{1,2,\ldots,n\}$ we have that 
\begin{itemize}
    \item if $\vartheta(q_i)=0$, then $\beta(v_i)=\frac{1}{8n}$,
    \item if $\vartheta(q_i)=1$, then $\beta(w_i)=\frac{1}{8n}$.
\end{itemize}
Similarly, it can be shown that for all $i\in\{1,2,\ldots,n\}$ and $j\in\{1,2,\ldots,r\}$ we have that 
\[
 \alpha(u)=\frac{1}{8},\quad \alpha(v)=\frac{1}{2},\quad \alpha(u_i)=\frac{1}{8n} 
\]
and 
\begin{itemize}
    \item if $\vartheta(q_i)=1$, then $\alpha(v_i)=\frac{1}{8n}$,
    \item if $\vartheta(q_i)=0$, then $\alpha(w_i)=\frac{1}{8n}$.
\end{itemize}
Finally, we have that $\alpha(x_j)\geq\frac{1}{8nr}$, which follows from the fact that there is at least one positively evaluated literal in $\psi_j$ and thus a vertex $z\in V'\cup W'$ such that $\alpha(z)=\frac{1}{8n}$ and $\kappa_A(z)(x_j)=\frac{1}{g(t)}\geq \frac{1}{r}$, where $t$ stands for a literal contained in $\psi$ (note that $\sum_{l=1}^r\alpha(x_l)=\alpha(u)=\frac{1}{8}$). It follows that the MR strategy profile $\pi$ achieves a frequency vector $\mu$ such that 
\[
 \mu(z) \ \geq \ \frac{1}{1-\zeta}\Obj(z)\ >\ \Obj(z)
\]  
for all $z\in V$, hence $\mu\geq Obj$. This finishes the proof of the `$\Rightarrow$' direction for MR profiles.

In the next paragraphs, we show how to modify the MR profile $\pi$ into a  FMR profile $\pi'$ achieving a frequency vector $\mu'\geq \Obj$. Let $\xi=(u,\kappa')$ be a FMR strategy where $\kappa'(x)(y)=\frac{1}{\texttt{deg}^+(x)}$ for all $x,y\in V$ such that $(x,y)\in E$, where $\texttt{deg}^+(x)$ stands for the outdegree of the vertex $x$ in $D$ ($\kappa'(x)(y)=0$ whenever $(x,y)\notin E$). Let $\lambda$ be the corresponding unique invariant distribution of the Markov chain induced by $\xi$. It is easy to observe that $\lambda(x)>0$ for all $x\in V$. Consider MR strategies $\xi_A'=(v,\kappa_A')$, $\xi_B'=(w,\kappa_B')$ where
\[
  \kappa_A'(x)(y)=\frac{(1-\zeta)\alpha(x)\kappa_A(x)(y)+\zeta\lambda(x)\kappa'(x)(y)}{(1-\zeta)\alpha(x)+\zeta\lambda(x)}
\]
and
\[
\kappa_B'(x)(y)=\frac{(1-\zeta)\beta(x)\kappa_B(x)(y)+\zeta\lambda(x)\kappa'(x)(y)}{(1-\zeta)\beta(x)+\zeta\lambda(x)}
\]
for all $x,y\in V$. We show that the strategy profile $\pi'=(\xi_A',\xi_B')$ is full and achieves a vector $\mu'\geq \Obj$. Let $(x,y)\in E$ be an arbitrary edge of $D$. It follows that 
\begin{eqnarray*}
 \kappa_A'(x)(y) &=& \frac{(1-\zeta)\alpha(x)\kappa_A(x)(y)+\zeta\lambda(x)\kappa'(x)(y)}{(1-\zeta)\alpha(x)+\zeta\lambda(x)}\\
  & \geq & \frac{\zeta\lambda(x)\kappa'(x)(y)}{1+\zeta\lambda(x)}\\
  & > & 0
\end{eqnarray*}
and 
\begin{eqnarray*}
\kappa_B'(x)(y) & = & \frac{(1-\zeta)\beta(x)\kappa_B(x)(y)+\zeta\lambda(x)\kappa'(x)(y)}{(1-\zeta)\beta(x)+\zeta\lambda(x)}\\
 &\geq&\frac{\zeta\lambda(x)\kappa'(x)(y)}{1+\zeta\lambda(x)}\\
 &>&0\,.
\end{eqnarray*} 
Hence, $\pi'$ is indeed a full MR strategy profile.

Let 
\begin{eqnarray*}
    \alpha'(x) & = & (1-\zeta)\alpha(x)+\zeta\lambda(x)\,,\\
    \beta'(x)  & = & (1-\zeta)\beta(x)+\zeta\lambda(x)
\end{eqnarray*}
for all $x\in V$. It is easy to see that $\alpha',\beta'$ are distributions over~$V$. In the rest of this paragraph, we prove that they are the unique invariant distributions for the Markov chains induced by strategies $\xi_A'$ and $\xi_B'$. Let $y\in V$ be an arbitrary vertex. We have that 
\begin{eqnarray*}
  & & \sum_{x\in V}\alpha'(x)\kappa_A'(x)(y)\\
 % & = & \sum_{x\in V}((1{-}\zeta)\alpha(x)+\zeta\lambda(x))\frac{(1{-}\zeta)\alpha(x)\kappa_A(x)(y)+\zeta\lambda(x)\kappa'(x)(y)}{(1-\zeta)\alpha(x)+\zeta\lambda(x)}\\
  &=& \sum_{x\in V}((1-\zeta)\alpha(x)\kappa_A(x)(y)+\zeta\lambda(x)\kappa'(x)(y))\\
  &=& (1-\zeta)\sum_{x\in V}\alpha(x)\kappa_A(x)(y)+\zeta\sum_{x\in V}\lambda(x)\kappa'(x)(y)\\
  &=& (1-\zeta)\alpha(y)+\zeta\lambda(y)=\alpha'(y)\,,
\end{eqnarray*}
implying that $\alpha'$ is an invariant distribution of the Markov chain induced by the strategy $\xi_A'$. In a similar way, it can be shown that $\beta'$ is an invariant distribution of the Markov chain induced by $\xi_B'$. Since both of the Markov chains are irreducible, it follows that their invariant distributions $\alpha'$, $\beta'$ are unique.

Let $\mu'$ be the frequency vector achieved by the strategy profile $\pi'$, %%$=(\xi_A',\xi_B')$
let $z\in V$. We have that 
\begin{eqnarray*}
  \mu'(z) & = & \alpha'(z)+\beta'(z)-\alpha'(z)\beta'(z)\\
  &=& \alpha'(z)+\beta'(z)(1-\alpha'(z))\\
  &\geq &\alpha'(z)
\end{eqnarray*}  
and 
\begin{eqnarray*}
 \mu'(z) & = & \alpha'(z)+\beta'(z)-\alpha'(z)\beta'(z)\\
 & =& \alpha'(z)(1-\beta'(z))+\beta'(z)\\
 &\geq& \beta'(z)\,.
\end{eqnarray*}
Furthermore, 
\begin{eqnarray*}
 \alpha'(z) & = & (1-\zeta)\alpha(z)+\zeta\lambda(z)\\
  &>& (1-\zeta)\alpha(z)\,,\\[1ex] 
 \beta'(z) & = & (1-\zeta)\beta(z)+\zeta\lambda(z)\\
  &>& (1-\zeta)\beta(z)\,.
\end{eqnarray*}  
In the previous paragraphs, we have actually shown that necessarily $\alpha(z)\geq \frac{1}{1-\zeta}Obj(z)$ or $\beta(z)\geq \frac{1}{1-\zeta}Obj(z)$, which implies that 
\begin{eqnarray*}
 \alpha'(z) & > & (1-\zeta)\alpha(z)\\
  &\geq& (1-\zeta)\frac{1}{1-\zeta}\Obj(z)\\
  & = & \Obj(z)
\end{eqnarray*}  
or 
\begin{eqnarray*}
 \beta'(z) & > & (1-\zeta)\beta(z)\\
  &\geq& (1-\zeta)\frac{1}{1-\zeta}\Obj(z)\\
  &=& \Obj(z)\,.
\end{eqnarray*}  
In either case, we obtain $\mu'(z)\geq \Obj(z)$, hence $\mu'\geq \Obj$.

To prove the `$\Leftarrow$' direction, let $\pi=(\xi_A,\xi_B)$ be a MR profile for two agents $A$ and $B$ achieving a frequency vector $\mu\geq Obj$. %%Since $\psi$ is equisatisfiable with $\varphi$, it is sufficient to prove that $\psi$ is satisfiable.
Consider the two Markov chains induced by the strategies $\xi_A$ and $\xi_B$. By definition, the relative frequencies have to approach $\mu$ in the limit with probability~$1$. Hence, we assume without restrictions that both $A$ and $B$ start in a BSCC and that there is only a single BSCC in either of the corresponding two Markov chains (this assumption is legitimate since $D$ is strongly connected). Let $\alpha$, $\beta$ be the unique invariant distributions corresponding to these two induced Markov chains ($\alpha$ belongs to $A$ and $\beta$ belongs to $B$). Since the achieved frequency vector is independent of the order of strategies in the profile, we assume wlog that $\alpha(v)\geq \beta(v)$, i.e., the agent $A$ visits $v$ at least as often as the agent $B$.

It is easy to see that at least one of the two agents has to use the self-loop on vertex $v$ with positive frequency (and thus positive probability in its MR strategy). Otherwise, none of the two agents is able to visit $v$ more often than in every fifth step, and the total frequency of visits to $v$ then cannot exceed 
\[
\frac{2}{5}<\frac{1-1/1024}{2}<\frac{1-1/(1024 n^4 r^2)}{2}=\frac{1-\zeta}{2}=\Obj(v),
\]
contradicting the assumption that $\mu\geq Obj$. Consequently, at least one of the two agents uses a strategy that induces an aperiodic Markov chain. The frequency of visits to each vertex $s\in V$ may thus be expressed as 
\begin{eqnarray*}
 \mu(s) & = & 1-(1-\alpha(s))(1-\beta(s))\\
 &=& \alpha(s)+\beta(s)-\alpha(s)\beta(s)\,.
\end{eqnarray*}
Observe that the vertex $u$ may be visited only when some of the vertices $x_j$ has been visited in the preceding step (with one possible exception at the very beginning of the run). It follows that 
\[
  \sum_{j=1}^{m}\mu(x_j) \ \geq \ \mu(u) \ \geq\  \Obj(u) \ =\ \frac{1-\zeta}{8}\,.
\] 
For the sake of contradiction, assume %% for the sake of contradiction
there is $s'\in V$ such that $\alpha(s')\beta(s')>2\zeta$. It follows that 
{\small
\begin{eqnarray*}
  2-2\zeta & = & \frac{1-\zeta}{8} + \frac{1-\zeta}{2} + \frac{7(1-\zeta)}{8} + n\frac{1-\zeta}{8n} + n\frac{1-\zeta}{8n}\\
  && + \ n\frac{1-\zeta}{8n} + \frac{1-\zeta}{8}\\
  & \leq &  \mu(u)+\mu(v)+\mu(w)+\sum_{i=1}^{n}\mu(u_i)\\
  && + \ \sum_{i=1}^{n}\mu(v_i)+\sum_{i=1}^{n}\mu(w_i)+\sum_{j=1}^{m}\mu(x_j)\\
  & = & \sum_{s\in V}\mu(s)\\
  & = & \sum_{s\in V}(\alpha(s)+\beta(s)-\alpha(s)\beta(s))\\
  & = & 1+1-\sum_{s\in V}\alpha(s)\beta(s)\leq 2-\alpha(s')\beta(s')\\
  & < & 2-2\zeta\,,
\end{eqnarray*}}%
which is a contradiction. Hence, $\alpha(s')\beta(s')\leq 2\zeta$ for all \mbox{$s'\in V$}. By applying this observation to $v$, we get 
\[
  \beta(v)^2 \ \leq\ \alpha(v)\beta(v) \ \leq \ 2\zeta
\]
and hence $\beta(v)\leq \sqrt{2\zeta}$. %%=\sqrt{\frac{2}{1024n^4r^2}}<\sqrt{\frac{4}{1024}}=\frac{1}{16}$.
Since there is only one edge leading from vertex $u$ and this edge leads to $v$, %% each visit to $u$ by agent $B$ is directly followed by a visit to $v$ by agent $B$, it holds that
we get $\beta(u)\leq \beta(v)$ and thus also $\beta(u)\leq \sqrt{2\zeta}$.

In this paragraph, we prove that for all $s\in V'\cup W'$, one of the inequalities
\begin{eqnarray*}
 \alpha(s) &<& 2\sqrt{\zeta}\,,\\
 \alpha(s) &>& \frac{1-\zeta}{8n}-2\sqrt{\zeta}
\end{eqnarray*}
holds. Note that these two inequalities cannot hold simultaneously as that would imply 
\begin{eqnarray*}
 \frac{1-\zeta}{8n}-2\sqrt{\zeta} &< &2\sqrt{\zeta},\\ 
 \frac{1-\zeta}{8n} &<& 4\sqrt{\zeta},\\
  1-\zeta&<&32n\sqrt{\zeta}
\end{eqnarray*}  
and therefore 
\begin{eqnarray*}
 \frac{1}{2} & = & 1-\frac{1}{2}\\
  &<& 1-\frac{1}{1024n^4r^2}\\
  &=& 1-\zeta<32n\sqrt{\zeta}\\
  &=& 32n\sqrt{\frac{1}{1024n^4r^2}}\\
  &=& \frac{1}{nr}\\
  &\leq&  \frac{1}{2}.
\end{eqnarray*} 
%% $\frac{1-\zeta}{8n}-2\sqrt{\zeta}<2\sqrt{\zeta}$, $\frac{1-\zeta}{8n}<4\sqrt{\zeta}$, $\frac{1}{2}=1-\frac{1}{2}<1-\frac{1}{1024n^4r^2}<\frac{1}{nr}=32n\frac{1}{32 n^2r}=32n\sqrt{\zeta}=32n\sqrt{\frac{1}{1024 n^4r^2}}=\frac{1}{nr}\leq \frac{1}{2}$).
Let $s\in V'\cup W'$. For the sake of contradiction, assume that 
\[
  2\sqrt{\zeta} \ \leq\ \alpha(s)\ \leq \ \frac{1-\zeta}{8n}-2\sqrt{\zeta}.
\]   
Since $\alpha(s)\beta(s)\leq 2\zeta$, we get $\beta(s)\leq \frac{2\zeta}{\alpha(s)}$ (notice that $\alpha(s)\neq 0$ because of $0<2\sqrt{\zeta}\leq \alpha(s)$). Recall that 
\[
\mu(s)\ =\ \alpha(s)+\beta(s)-\alpha(s)\beta(s) \ \geq \ \frac{1-\zeta}{8n}\ =\ \Obj(s).
\] 
From this, a contradiction follows easily:
\begin{eqnarray*}
  \frac{1-\zeta}{8n} & \leq &  \alpha(s)+\beta(s)-\alpha(s)\beta(s)\\
  &=& \alpha(s)+(1-\alpha(s))\beta(s)\\
  &\leq&  \alpha(s)+(1-\alpha(s))\frac{2\zeta}{\alpha(s)}\\
  &\leq& \bigg(\frac{1-\zeta}{8n}-2\sqrt{\zeta}\bigg)+\bigg(\frac{2\zeta}{\alpha(s)}-2\zeta\bigg)\\
  &\leq& \frac{1-\zeta}{8n}-2\sqrt{\zeta}+\frac{2\zeta}{2\sqrt{\zeta}}-2\zeta\\
  &=&\frac{1-\zeta}{8n}-\sqrt{\zeta}-2\zeta<\frac{1-\zeta}{8n}\,.
\end{eqnarray*}

Similarly, one can prove that for all $s\in V'\cup W'$, precisely one of the inequalities 
\begin{eqnarray*}
 \beta(s) &<& 2\sqrt{\zeta}\,,\\
\beta(s) &>& \frac{1-\zeta}{8n}-2\sqrt{\zeta}
\end{eqnarray*}
holds.

%% as in the previous paragraph.
Let $i\in\{1,2,\ldots,n\}$. Observe that $\beta(u_i)\leq\beta(u)$ (the long-run frequency of agent $B$ visiting $u_i$ cannot be greater than its frequency of visiting $u$ since $B$ has to visit $u$ between any two consecutive visits to $u_i$), hence $\beta(u_i)\leq \sqrt{2\zeta}$. Since 
\[
 \frac{1-\zeta}{8n}=\Obj(u_i)\leq\mu(u_i)\leq\alpha(u_i)+\beta(u_i)\leq \alpha(u_i)+\sqrt{2\zeta},
\] 
we get $\alpha(u_i)\geq \frac{1-\zeta}{8n}-\sqrt{2\zeta}$. Since 
\[
 \alpha(v_i)+\alpha(w_i)\geq\alpha(u_i)\geq \frac{1-\zeta}{8n}-\sqrt{2\zeta},
\] 
it follows that at least one %% (exactly one, as follows from the next paragraph)
of the following two inequalities must hold: 
\begin{eqnarray*}
  \alpha(v_i) & \geq & \frac{1}{2}\left(\frac{1-\zeta}{8n}-\sqrt{2\zeta}\right)\,,\\ %%$=\frac{1-\zeta}{16n}-\frac{\sqrt{2\zeta}}{2}$ %% ($\geq 2\sqrt{\zeta}$)
  \alpha(w_i) & \geq & \frac{1}{2}\left(\frac{1-\zeta}{8n}-\sqrt{2\zeta}\right)\,,
\end{eqnarray*}  
where %%$=\frac{1-\zeta}{16n}-\frac{\sqrt{2\zeta}}{2}$.
$\frac{1}{2}(\frac{1-\zeta}{8n}-\sqrt{2\zeta})\geq 2\sqrt{\zeta}$, which can be proved as follows: 
\begin{eqnarray*}
 1-\zeta &=& 1-\frac{1}{1024 n^4 r^2}\\
  &\geq& 1-\frac{1}{1024} \geq \frac{3}{4} \geq \frac{3}{2nr}\\
  &=& 48n\frac{1}{32 n^2r}=48n\sqrt{\zeta}
\end{eqnarray*}  
and hence  
\begin{eqnarray*}
 \frac{1}{2}\left(\frac{1-\zeta}{8n}-\sqrt{2\zeta}\right) &\geq& \frac{1}{2}\left(\frac{48n\sqrt{\zeta}}{8n}-\sqrt{2}\sqrt{\zeta}\right)\\
 & \geq& 2\sqrt{\zeta}.
\end{eqnarray*} 
By combining this with the previous statements, we get that at least one of the two inequalities $\alpha(v_i)>\frac{1-\zeta}{8n}-2\sqrt{\zeta}$, $\alpha(w_i)>\frac{1-\zeta}{8n}-2\sqrt{\zeta}$ must hold.

Since these inequalities hold for all $i\in\{1,2,\ldots,n\}$, there are at least $n$ vertices $s\in V'\cup W'$ such that 
\[
  \alpha(s)>\frac{1-\zeta}{8n}-2\sqrt{\zeta}\,.
\]   
For the sake of contradiction, assume that there are at least $n+1$ such vertices. We have shown that $\beta(u)\leq \sqrt{2\zeta}$. Since $\alpha(u)+\beta(u)\geq \mu(u)\geq\frac{1-\zeta}{8}$, we obtain
\[
  \alpha(u)\ \geq\ \frac{1-\zeta}{8}-\beta(u)\ \geq\ \frac{1-\zeta}{8}-\sqrt{2\zeta}\,.
\]  
Similarly, $\beta(v)\leq \sqrt{2\zeta}$ and $\alpha(v)+\beta(v)\geq \mu(v)\geq\frac{1-\zeta}{2}$, hence 
\[
  \alpha(v)\ \geq\ \frac{1-\zeta}{2}-\beta(v)\ \geq\ \frac{1-\zeta}{2}-\sqrt{2\zeta}\,.
\]  
We obtain 
{\small\begin{eqnarray*}
 1&=&\sum_{s\in V}\alpha(s)\\
 &\geq& \alpha(u)+\alpha(v)+(n+1)\left(\frac{1-\zeta}{8n}-2\sqrt{\zeta}\right)\\
 && +\ \sum_{j=1}^r\alpha(x_j)+\sum_{i=1}^n\alpha(u_i)\\
  &=&  \alpha(u)+\alpha(v)+(n+1)\left(\frac{1-\zeta}{8n}-2\sqrt{\zeta}\right)+\alpha(u)+\alpha(u)\\
  &=& 3\alpha(u)+\alpha(v)+(n+1)\left(\frac{1-\zeta}{8n}-2\sqrt{\zeta}\right)\\
  &\geq& 3\left(\frac{1-\zeta}{8}-\sqrt{2\zeta}\right)+\left(\frac{1-\zeta}{2}-\sqrt{2\zeta}\right)\\
  && +\ (n+1)\left(\frac{1-\zeta}{8n}-2\sqrt{\zeta}\right)\\
  &=& (1+\frac{1}{8n})(1-\zeta)-(2n+2+4\sqrt{2})\sqrt{\zeta}\\
  &=& 1+\frac{1}{8n}-\frac{1}{1024 n^4 r^2}\\
  &&-\ \frac{1}{8192 n^5r^2}-(2n+2+4\sqrt{2})\frac{1}{32 n^2 r}\\
  &>& 1+\frac{1}{8n}-\frac{1}{1024 n}-\frac{1}{8192 n}-\frac{1}{32 n}-\frac{1}{16 n}\\
  &>&1\,,
\end{eqnarray*}}%  
which is a contradiction. Hence, there are exactly $n$ vertices $s\in V'\cup W'$ such that 
\[
\alpha(s)>\frac{1-\zeta}{8n}-2\sqrt{\zeta}\,.
\]
It follows that for the remaining $n$ vertices $s\in V'\cup W'$ we have that $\alpha(s)<2\sqrt{\zeta}$. Combining this with the previous statements, %% (that is, $\alpha(v_i)>\frac{1-\zeta}{8n}-2\sqrt{\zeta}$ or $\alpha(w_i)>\frac{1-\zeta}{8n}-2\sqrt{\zeta}$ for all $i\in\{1,2,\ldots,n\}$),
we get that for all $i\in\{1,2,\ldots,n\}$, it holds either that $\alpha(v_i)> \frac{1-\zeta}{8n}-2\sqrt{\zeta}$ and $\alpha(w_i)<2\sqrt{\zeta}$ or that $\alpha(w_i)> \frac{1-\zeta}{8n}-2\sqrt{\zeta}$ and $\alpha(v_i)<2\sqrt{\zeta}$.

Consider a valuation $\eta$ such that $\eta(q_i)=0$ (meaning that $q_i$ evaluates to false) if and only if $\alpha(v_i)<2\sqrt{\zeta}$ (for all $i\in\{1,2,\ldots,n\}$). We are to prove that $\eta(\psi)=1$ (meaning that $\psi$ is satisfied under the valuation $\eta$). %% It follows from the previously proved statements that
Observe that for any $i\in \{1,2,\ldots,n\}$ such that $\eta(\neg q_i)=0$, we have that $\alpha(v_i)\geq 2\sqrt{\zeta}$ and thus (according to the previously proved statements) also $\alpha(v_i)> \frac{1-\zeta}{8n}-2\sqrt{\zeta}$ and $\alpha(w_i)<2\sqrt{\zeta}$. For the sake of contradiction, assume $\eta(\psi)=0$. Let $\psi_j$ be a clause of $\psi$ such that all of the literals in $\psi_j$ evaluate to $0$. Let $(z,x_j)\in E$ be an arbitrary edge leading to $x_j$. It follows from the construction of the graph $D$ that necessarily $z=v_l$ or $z=w_l$ for some $l\in\{1,2,\ldots,n\}$. In the case when $z=v_l$, it follows that $\psi_j$ contains the positive literal $q_l$, therefore $\eta(q_l)=0$ and $\alpha(z)=\alpha(v_l)<2\sqrt{\zeta}$. In the case when $z=w_l$, it follows that $\psi_j$ contains the negative literal $q_l$, therefore $\eta(\neg q_l)=0$ and $\alpha(z)=\alpha(w_l)<2\sqrt{\zeta}$. We thus obtain $\alpha(z)<2\sqrt{\zeta}$ in both cases. Note that the clause $\psi_j$ contains at most $n$ literals because there are no tautological clauses in $\psi$. It follows that there are at most $n$ edges leading to $x_j$ (where for each such edge $(z,x_j)$ we have that $\alpha(z)<2\sqrt{\zeta}$) and thus $\alpha(x_j)\leq 2n\sqrt{\zeta}$. It also holds that 
\[
 \alpha(x_j)+\beta(x_j)\geq \mu(x_j)\geq Obj(x_j)=\frac{1-\zeta}{8nr}
\] 
and $\beta(x_j)\leq \beta(u)\leq \sqrt{2\zeta}$ (where $\beta(x_j)\leq \beta(u)$ is a trivial observation and we have already shown that $\beta(u)\leq \sqrt{2\zeta}$), allowing us to derive 
\[
 \frac{1-\zeta}{8nr}-\sqrt{2\zeta}\leq (\alpha(x_j)+\beta(x_j))-\beta(x_j)=\alpha(x_j)\leq 2n\sqrt{\zeta}\,,
\] implying that 
\[
 1-\zeta-8nr\sqrt{2\zeta}\leq 16n^2r\sqrt{\zeta}
\] 
and thus finally %% where substituting the definition of $\zeta$ leads to $1-\frac{1}{1024 n^4 r^2}-8nr\sqrt{\frac{2}{1024 n^4 r^2}}\leq 16n^2r\sqrt{\frac{1}{1024 n^4 r^2}}$.
\[  
\frac{1}{2}<1-\frac{1}{1024}-\frac{\sqrt{2}}{4}\leq 1-\zeta-8nr\sqrt{2\zeta}\leq 16n^2r\sqrt{\zeta}=\frac{1}{2}\,,
\]
which is a contradiction. We have thus shown that $\eta(\psi)=1$, meaning that $\psi$ is satisfiable. This finishes the proof of the `$\Leftarrow$' direction in the case of MR strategy profiles. The result for FMR strategy profiles is a trivial consequence, since every FMR profile is a MR profile.

The presented polynomial-time reduction may be generalized so that it proves \NP-hardness also for FR$_m$ profiles (for any fixed $m\in\mathbb{N}_+$). The extension is somewhat technical, but it does not require any substantially new ideas.

\subsection{A Proof of Theorem 3}

We start by introducing an auxiliary decision problem \mbox{\textit{mod-SAT}} and proving its \NP-completeness. 

An instance of mod-SAT is a list of ordered pairs $(n_1,S_1), (n_2,S_2), \ldots\, (n_d,S_d)$, where each $n_i$ is a positive integer and each $S_i$ is a subset of $\{0,1,\ldots,n_i{-}1\}$. The sets are represented as lists of their elements, all numeric values are encoded in unary (or alternatively binary), and the number of input ordered pairs (denoted by $d$) is a nonnegative integer. The question is whether there exists an integer $x$ such that $(x\texttt{ mod }n_i)\in S_i$ for each $i\in\{1,2,\ldots,d\}$. % where $\texttt{mod}$ stays for the modulo operation.

%% In the rest of this section,
%% We show that mod-SAT is NP-complete.
Clearly, mod-SAT belongs to \NP, because the size of the witnessing integer $x$ can be bounded by $\prod_{i=1}^{d}n_i$, and hence the length of the binary encoding of $x$ is polynomial in the size of the input instance.

% since if any integer $x$ witnesses that the input instance of the problem is positive, then there exists also such a witness which is nonnegative and less than $\prod_{i=1}^{d}n_i$, where the length of the binary representation of this product is polynomial in the size of the input (and the modulo operation is computable in polynomial time even for binary input values).

We prove \NP-hardness of mod-SAT by a polynomial-time reduction from 3-SAT. 
An instance of 3-SAT is a formula in conjunctive normal form where each clause contains precisely three literals, and the question is whether the formula is satisfiable.  
%
% 3-SAT is a well- known NP-complete problem: a propositional formula in conjunctive normal form with exactly three literals %%\footnote{A literal is either a propositional variable or its negation.}
% in each clause is given on the input and the task is to decide whether such formula is satisfiable (some definitions only require each clause of the input formula to contain at most three literals, however, such a formula can always be efficiently transformed to an equisatisfiable one containing exactly three literals in each of its clauses). 
%
A pair of literals is \textit{conflicting} if one of them is a propositional variable and the other is a negation of the same variable. %% We denote by $l'\not\perp l''$ that literals $l'$ and $l''$ are not conflicting.

Consider an instance $\psi_2 \land \cdots \land \psi_{n+1}$ of 3-SAT, where each $\psi_i$ is a clause of the form $(l_{i,0}\lor l_{i,1}\lor l_{i,2})$, 
%% $(l_{2,0}\lor l_{2,1}\lor l_{2,2})\land(l_{3,0}\lor l_{3,1}\lor l_{3,2})\land(l_{4,0}\lor l_{4,1}\lor l_{4,2})\land\ldots\land(l_{n+1,0}\lor l_{n+1,1}\lor l_{n+1,2})$,
where each $l_{i,j}$ is a literal and $n\geq 1$ is the number of clauses (the first clause is intentionally denoted by $\psi_2$). %% and $n>0$ is the number of clauses.
In the next paragraphs, we describe the polynomial-time reduction of 3-SAT to mod-SAT.

The reduction starts by generating the first $n$ odd prime numbers $p_2,\ldots,p_{n+1}$. Recall that this is achievable in polynomial time, and the size of $p_{n+1}$ is asymptotically bounded by $n \log(n)$.
%
% are generated: this may be achieved in polynomial time, for example using the sieve of Eratosthenes, starting with a list of numbers $2,3,4,\ldots,(n+1)^2$. There are known bounds on the prime-counting function $\pi$ which imply that $p_{n+1}<(n+1)^2$ for all $n\geq 1$. Specifically, according to Theorem 29 of [??], it is true for all $x\geq 55$ that $\frac{x}{\log{x}+2}<\pi(x)$ ($\log$ stands for the natural logarithm). Using the fact that for all $x\geq 55$ it holds that $\log{x}+2\leq \sqrt{x}$, we can further derive that $p_{\lfloor\sqrt{x}\rfloor}=p_{\lfloor\frac{x}{\sqrt{x}}\rfloor}\leq p_{\lfloor\frac{x}{\log{x}+2}\rfloor}<x$, substituting $(n+1)^2$ for $x$ then leads to $p_{n+1}<(n+1)^2$ for all $n$ such that $(n+1)^2\geq 55$, that is, for all $n\geq 7$. The cases for $1\leq n\leq 6$ can be easily verified.%%% deriving the bound is then an easy exercise, checking several small values of $n$ separately and using the fact that for all $x\geq 55$ it holds that $\log{x}+2\leq \sqrt{x}$.
%
The constructed instance of mod-SAT then contains 
\begin{itemize}
    \item an ordered pair $(p_i, \{0,1,2\})$ for all $i\in\{2,\ldots,n{+}1\}$,
    \item an ordered pair $(p_j \cdot p_k, A_{j,k})$ for each pair $\psi_j$, $\psi_k$ of distinct clauses, where $A_{j,k}$ consists of all $m$ such that 
    \begin{itemize}
        \item $0 \leq m < p_j \cdot p_k$, 
        \item $m\texttt{ mod }p_j \leq 2$, 
        \item $m\texttt{ mod }p_k \leq 2$, 
        \item the literals $l_{j,m\texttt{ mod }p_j}$ and $l_{k,m\texttt{ mod }p_k}$ are not conflicting.
    \end{itemize}
\end{itemize}
Observe that the size of the above instance is polynomial in the size of the considered 3-SAT instance, even if all numerical constants are encoded in unary.

% for each $i\in\{2,3,4,\ldots,n+1\}$ an ordered pair $(p_i, \{0,1,2\})$ and for each pair of two distinct clauses $\psi_j$, $\psi_k$ an ordered pair %%$(p_jp_k, \{m\in\{0,1,\ldots,p_jp_k-1\} : l_{j,m\texttt{ mod }p_j}\text{ and }l_{k,m\texttt{ mod }p_k}\text{ are not conflicting}\})$.
% $(p_jp_k, \{m\in\mathbb{N}_{p_jp_k} : l_{j,m\texttt{ mod }p_j}\not\perp l_{k,m\texttt{ mod }p_k}\})$, where $\mathbb{N}_{p_jp_k}=\{0,1,\ldots,p_jp_k-1\}$, where $\not\perp$ denotes that the two literals are not conflicting and where $m$ is not included to the set if an invalid reference to a literal is produced. %% (if $m\texttt{ mod }p_j>2$ or $m\texttt{ mod }p_k>2$).
% It should be rather obvious that thanks to the existing bounds on the primes it is possible to construct such an input in polynomial time even using unary encoding.

It remains to show that the constructed list of ordered pairs is a positive instance of mod-SAT if and only if the original propositional formula is satisfiable. The two implications are proven separately. 
%Observe that a propositional formula in conjunctive normal form is satisfiable if and only if it is possible to choose one literal from each clause so that the chosen literals are pairwise non-conflicting.

Assume that the constructed list of ordered pairs is a positive instance of mod-SAT. We need to show that the original formula is satisfiable, i.e., it is possible to choose one literal from each clause so that all chosen literals are pairwise non-conflicting.
Let $x$ be an integer witnessing that the constructed mod-SAT instance is positive. For each $i\in \{2,3,4,\ldots,n+1\}$, choose the literal $l_{i,x\texttt{ mod }p_i}$ from the clause $\psi_i$. We show that these literals are pairwise non-conflicting. Let $\psi_j$, $\psi_k$ be distinct clauses. The corresponding chosen literals are then $l_{j,x\texttt{ mod }p_j}$ and $l_{k,x\texttt{ mod }p_k}$. These literals are non-conflicting by our construction of ordered pairs.

% Since %% $(x\texttt{ mod }p_jp_k)\in \{m\in\{0,1,\ldots,p_jp_k-1\} : l_{j,m\texttt{ mod }p_j}\text{ and }l_{k,m\texttt{ mod }p_k}\text{ are not conflicting}\}$,

% $(x\texttt{ mod }p_jp_k)\in\{m\in\mathbb{N}_{p_jp_k} : l_{j,m\texttt{ mod }p_j}\not\perp l_{k,m\texttt{ mod }p_k}\}$, it follows that $l_{j,(x\texttt{ mod }p_jp_k)\texttt{ mod }p_j}=l_{j,x\texttt{ mod }p_j}$ is not conflicting with $l_{k,(x\texttt{ mod }p_jp_k)\texttt{ mod }p_k}=l_{k,x\texttt{ mod }p_k}$.

Conversely, assume that the original formula is satisfiable. We choose a literal
$l_{i,m_i}$ from each clause $\psi_i$ so that all of them are pairwise non-conflicting.
%
% Our task is now to show that the constructed list of ordered pairs is a positive instance of mod-SAT. Let a literal $l_{i,m_i}$ be chosen from the clause $\psi_i$ for each $i\in \{2,3,4,\ldots,n+1\}$ in a way that none of the chosen literals are conflicting.
%
Let $x$ be an integer such that $(x\texttt{ mod }p_i) = m_i$ for all $i\in \{2,3,4,\ldots,n+1\}$. Observe that $x$ always exists due to the Chinese remainder theorem. 
It is easy to verify that $x$ witnesses the positivity of the constructed mod-SAT instance.
%  is positive, it remains to show that for any pair of two distinct clauses $\psi_j$, $\psi_k$ it holds that $(x\texttt{ mod }p_jp_k) \in \{m\in\mathbb{N}_{p_jp_k} : l_{j,m\texttt{ mod }p_j}\not\perp l_{k,m\texttt{ mod }p_k}\}$, which means to show that $l_{j,(x\texttt{ mod }p_jp_k)\texttt{ mod }p_j} = l_{j,x\texttt{ mod }p_j} = l_{j,m_j}$ is not conflicting with $l_{k,(x\texttt{ mod }p_jp_k)\texttt{ mod }p_k} = l_{k,x\texttt{ mod }p_k} = l_{k,m_k}$, which is an assumption. This finishes the proof that mod-SAT is NP-complete.

Now we can continue with the proof of Theorem~3. We show that for a given strongly connected graph $D$, a vertex $v$ of $D$, and a MR profile $\pi$, the problem whether $v$ is visited with frequency~$1$ is \coNP-hard. We reduce mod-SAT to the complement of this problem.

% (for a variable number of agents) %%for $k \geq 1$ agents
% it is a coNP-hard problem to decide whether vertex $v$ is visited with frequency $1$. 
%%\begin{theorem} Let the input consist of a directed graph $(V,E)$, initial positions of agents, a memoryless randomized (deterministic) strategy profile and a vertex $v\in V$. The problem to decide whether the frequency of visits to $v$ is equal to $1$ is co-NP-hard. \end{theorem}
%%\begin{theorem}
%%Let $D$ be a graph, $\gamma$ a trivial coloring, $v$ a vertex of $D$, and $\pi$ a MR profile %%for $k \geq 1$ agents
%%such that $\pi$ achieves some (unknown) frequency vector $\mu$. The problem whether $\mu(v)=1$ is \mbox{\coNP-hard}.
%%\end{theorem}
% We prove this result by reducing the previously defined NP-complete problem mod-SAT to the complement of this problem (that is, to the problem of deciding whether the frequency of visits to $v$ is less than $1$) using a polynomial-time reduction.

Let $(n_1,S_1), (n_2,S_2), \ldots\, (n_d,S_d)$ be an instance of mod-SAT with all numeric values encoded in unary. Without restrictions, we assume $d\geq 1$ and $n_i\geq 2$ for all $i$. The reduction constructs a graph $D=(V,E,\emptyset)$, where
\begin{itemize}
    \item $V=\{v\}\cup\bigcup_{i=1}^d\{v_{i,1}, \ldots, v_{i,n_i-1}\}$,
    \item $E$ consists of the edges 
    \begin{itemize}
        \item $(v,v_{j,n_j-1})$,
        \item $(v_{j,n_j-k}, v_{j,n_j-k-1})$ for all $k \in \{1,\ldots,n_j-2\}$,
        \item $(v_{j,1}, v)$,
    \end{itemize}  
    for all $j \in \{1,\ldots,d\}$.
\end{itemize}

% $V=\{v\}\cup\bigcup_{i=1}^d\{v_{i,1}, v_{i,2}, \ldots, v_{i,n_i-1}\}$ and where %%$E=\bigcup_{j=1}^d \{(v,v_{j,n_j-1}),(v_{j,n_j-1}, v_{j,n_j-2}),(v_{j,n_j-2}, v_{j,n_j-3}), \ldots, (v_{j,3}, v_{j,2}), (v_{j,2}, v_{j,1}), (v_{j,1}, v)\}$.
% $E=\bigcup_{j=1}^d \{(v,v_{j,n_j-1}),(v_{j,n_j-1}, v_{j,n_j-2}),(v_{j,n_j-2}, v_{j,n_j-3}),\ldots,$ $(v_{j,3}, v_{j,2}), (v_{j,2}, v_{j,1}), (v_{j,1}, v)\}$. %% $E_j=\{(v,v_{j,n_j-1}),(v_{j,n_j-1}, v_{j,n_j-2}),(v_{j,n_j-2}, v_{j,n_j-3}), \ldots, (v_{j,3}, v_{j,2}), (v_{j,2}, v_{j,1}), (v_{j,1}, v)\}$ for each $j\in \{1,2,\ldots,d\}$. , where $E=\bigcup_{j=1}^d E_j$.
Observe that the constructed directed graph $D$ consists of $d$ directed cycles, where all these cycles share a common vertex $v$ (no other vertex is shared).

Now we define a MR profile for $D$. 
% Now we describe an MR profile, that is, the number, initial positions and memoryless strategies of all agents. %%At the beginning, no agent is assigned its starting vertex and strategy.
Observe that an agent can make a non-trivial decision about the next vertex only in the vertex $v$ (all other vertices have only one outgoing edge). Hence, a MR strategy profile is fully described by the initial positions of all agents and their behavior in~$v$.
% to which vertex to go next only when currently being in vertex $v$ (otherwise, there is only one possible successor vertex). In order to provide a complete description of the strategy profile, it is thus sufficient to prescribe for each agent its initial position and how the agent behaves in vertex $v$. 
For each $i\in \{1,2,\ldots,d\}$, we add $n_i-\left|S_i\right|$ new agents (later called ``agents introduced in the $i$-th iteration''), each of them always deterministically continuing to vertex $v_{i,n_i-1}$ when being in $v$. Initially, we place one of those agents to $v$ if $0\notin S_i$ and, for every $j\in \{1,2,\ldots,n_i-1\}$ such that $j\notin S_i$, we place one of those agents to $v_{i,j}$. Hence, the total number of agents is $k=\sum_{i=1}^d(n_i-\left|S_i\right|)$.
% It should be rather obvious that the whole construction can be implemented by an algorithm running in time polynomial in the length of the considered input instance of mod-SAT (using unary encoding).

It remains to show that the frequency of visits to $v$ is less than $1$ if and only if the considered input instance of mod-SAT is positive. Note that all agents deterministically walk around the cycle where they started. Therefore, the arrangement of all agents periodically repeats with a finite period (not exceeding the least common multiple of the lengths of the cycles). Hence, the frequency of visits to $v$ is less than $1$ if and only if there is a reachable arrangement such that none of the agents is in $v$. 
%It is a rather simple observation of the construction that the following statement is true for all $x\in\mathbb{N}$ and all $i\in \{1,2,\ldots,d\}$: none of the agents introduced in the $i$-th iteration is present in $v$ after $x$ steps from the beginning if and only if $(x\texttt{ mod }n_i)\in S_i$.

Assume that the frequency of visits to $v$ is less than $1$. Let $x\in\mathbb{N}$ be such that after $x$ steps from the beginning, none of the agents visits $v$. Hence, for each $i\in \{1,2,\ldots,d\}$, none of the agents introduced in the $i$-th iteration is present in $v$ after $x$ steps from the beginning, i.e., $(x\texttt{ mod }n_i)\in S_i$. It follows that the considered input instance of mod-SAT is positive.%% It follows that none out of all the $k$ agents is present in $v$ after $x$ time steps only if $x$ is the witness that the considered input instance of mod-SAT is positive.

Now assume that the considered instance of mod-SAT is positive. Let $x\in\mathbb{N}$ be such that $(x\texttt{ mod }n_i)\in S_i$ for each $i\in\{1,2,\ldots,d\}$. It follows that, for each $i$, none of the agents introduced in the $i$-th iteration is present in $v$ after $x$ steps from the beginning, which means that there is no agent in $v$ after $x$ steps from the beginning, and hence the frequency of visits to $v$ is less than $1$.%% It again follows from the construction that after $x$ time steps from the beginning there is none of the agents introduced in the $i$-th step present in $v$ if and only if

This concludes the proof of Theorem 3. Note that we have actually proved the hardness result for memoryless \emph{deterministic} profiles (a proper subclass of memoryless randomized profiles). Using a modification of the described polynomial-time reduction, it can also be proved, for an arbitrary fixed rational number $r\in\left(0,1\right]$, that the problem of deciding whether $v$ is visited with frequency $r$ (at least $r$, respectively) is \coNP-hard.

\subsection{A Proof of Theorem 4}

Let us fix $k \geq 1$.
Let $D = (V,E,p)$ be an MDP, $\Col : V \to \cols$ a coloring, $\Obj$ a frequency vector, and $m \geq 1$. We show that the problem whether there exists a FR$_m$ memory profile $\pi$ for $k$ agents achieving a frequency vector $\mu \geq \Obj$ is decidable in polynomial space, assuming that $m$ is encoded in unary.

As observed in Section~2 in the main body of the paper, the construction of a FR$_m$ strategy profile for $D$ is equivalent to the construction of a MR profile for an MDP $D'$ obtained from $D$ by augmenting vertices with memory states. Since the number of vertices of $D'$ is $|V| \times m$, the increase in size is \emph{linear in~$m$}. Hence, it actually suffices to prove that the existence of \emph{MR profile} achieving $\mu \geq \Obj$ is in \PSPACE. We demonstrate this by designing a non-deterministic polynomial-space decision algorithm.

Recall the notion of an end component introduced in Section~\ref{sec-MDP-normal-form}. Also, recall that for an arbitrary strategy, almost all runs of $D$ eventually stay in some end component and execute all edges of this end component infinitely often.

To decide the existence of a suitable MR profile $\pi = (\sigma_1,\ldots,\sigma_k)$, the algorithm starts by guessing, for every $i \in \{1,\ldots,k\}$, an end component $D_i = (V_i,E_i,p_i)$ where $\sigma_i$ stays, together with an initial vertex $v_i \in V_i$ of $\sigma_i$.  Note that $\sigma_i$ is in fact a \emph{full} MR strategy for $D_i$. Now the algorithm computes a formula $\Phi$ of the existential fragment of first order theory of the reals which states the existence of suitable positive values for the variables representing the edge probabilities such that the induced frequency vector $\mu$ satisfies $\mu \geq \Obj$. The subformula encoding the vector $\mu$ is constructed as follows. First, the algorithm computes the period $d_i$ of the Markov chain induced by $\sigma_i$ and $D_i$ for every $i \in \{1,\ldots,k\}$. Then, it computes the least common multiple $d$ of all $d_i$ and constructs a subformula encoding $\mu(c)$ for every $c \in \cols$. This formula is similar to the expression~(2) in the main body of the paper, i.e.,

\begin{equation*}
    \mu(c) = \frac{1}{d} \sum_{j=0}^{d-1}\bigg(
       1 - \prod_{i=1}^k \bigg( 1 - d_i \cdot\sum_{v \in V^c(i,j)} \Inv_i(v) \bigg)\bigg)\,.
\end{equation*}
The difference is that $d$ now represents the least common multiple of all $d_i$. Since $d_i \leq |V|$ for every $i \in \{1,\ldots,k\}$, the size of the above expression is \emph{polynomial for every fixed $k$} (although the degree of the polynomial grows \emph{exponentially} in~$k$).

Since the size of the resulting formula $\Phi$ is polynomial and $\Phi$ belongs to the \emph{existential fragment} of first order theory of the reals, the validity of $\Phi$ in decidable in polynomial space. Hence, the whole non-deterministic algorithm deciding the existence of a suitable MR profile $\pi$ achieving a frequency vector $\mu \geq \Obj$ runs in polynomial space. 

\section{MDPs in Normal Form}
\label{sec-MDP-normal-form}

In this section, we show that for purposes of steady-state synthesis, we can safely assume that MDPs are given in the normal form defined in Section~4 in the main body of the paper. 

\begin{definition}
Let $D = (V,E,p)$ be an MDP. An \emph{end component} of $D$ is a triple $D' = (V',E',p')$ where $V' \subseteq V$, $E' \subseteq E \cap (V' {\times} V')$, and $p'$ is the restriction of $p$ to $V' \cap V_S$ such that 
\begin{itemize}
    \item for every $v \in V'$, there is an outgoing edge $(v,u) \in E'$;
    \item if $v \in V_S \cap V'$  and $(v,u) \in E$, then $(v,u) \in E'$;
    \item $(V',E')$ is strongly connected.
\end{itemize}
An end component is \emph{maximal (a MEC)} if it is maximal w.r.t.{} component-wise inclusion.
\end{definition}

Every MDP $D$ with $m$ vertices can be efficiently decomposed into at most $m$ pairwise disjoint MECs $D_1,\ldots,D_m$, and each of these MECs can be seen as a strongly connected MDP. 

For every run $\omega$ in $D$, let 
\begin{itemize}
    \item $V_{\omega}$ be the set of all $v \in V$ that occur infinitely often along $\omega$;
    \item $E_\omega$ be the set of all edges that occur infinitely often along $\omega$. 
\end{itemize}
%A pair $(V',E')$ such that $(V',E') = (V_\omega,E_\omega)$ for some run $\omega$ of $D$ is \emph{eligible}. 
%Let  $(V_1,E_1),\ldots,(V_n,E_n)$ be all eligible pairs.
For all $V' \subseteq V$ and $E' \subseteq E$, let $\Run(V',E')$ be the set of all runs $\omega$ of $D$ such that $(V_\omega,E_\omega) = (V',E')$.

\begin{proposition}
\label{prop-MEC}
    Let $\xi$ be a HR strategy, $V' \subseteq V$, and $E' \subseteq E$. If $\prob_\xi(\Run(V',E')) > 0$, then $(V',E',p')$ is an end component of $D$, where $p'$ is the restriction of $p$ to $V' \cap V_S$.
\end{proposition}
\begin{proof}
    Let $\omega = v_1,v_2,\ldots$ be a run of $\Run(V',E')$. Suppose $v \in V'$, Since $v$ occurs infinitely often in $\omega$, some outgoing edge $(v,u)$ of $v$ occurs infinitely often along $\omega$, which implies $(v,u) \in E'$. Also observe that if $v,u \in V'$, then $\omega$ contains a finite path from $v$ to $u$. Hence $(V',E')$ is strongly connected. Now suppose $v \in V' \cap V_S$ and $(v,u) \in E$. Then, the $\prob_\xi$ probability of all runs $\omega$ such that $v$ occurs infinitely often in $\omega$ but $(v,u)$ occurs only finitely often in $\omega$ is zero. Since $\prob_\xi(\Run(V',E')) > 0$, we have that $(v,u)$ occurs infinitely often in almost all runs of $\Run(V',E')$. Hence, $(v,u) \in E'$. This implies that $(V',E',p')$ is an end component.% Since $\prob_\xi(\Run(V',E')) > 0$, we have the $(v,u)$ occurs infinitely often in almost all runs of $\Run(V',E')$. Hence, $(v,u) \in E$. This implies that $(V',E')$ is an end component.
\end{proof}

According to Proposition~\ref{prop-MEC}, almost all runs eventually stay in some end component, and hence also in some MEC. 

Now consider a strategy profile $\pi = (\xi_1,\ldots,\xi_k)$ such that $\pi$ achieves some frequency vector~$\mu$. Let $\calA$ be a function assigning to every $i \in \{1,\ldots,k\}$ a pair $(D^i,v^i)$ where $D^i$ is a MEC of $D$ and $v^i$ is vertex of $D^i$. Furthermore, let $\MRun_{\calA}$ be the set of all multiruns $(\omega_1,\ldots,\omega_k)$ such that, for all $i \in \{1,\ldots,k\}$, we have that $\omega_i$ stays in the MEC $D^i$ and the first vertex of $D^i$ visited by $\omega_i$ is $v^i$. We say that $\calA$ is \emph{$\pi$-eligible} if $\prob_\pi(\MRun_{\calA}) > 0$.

Since $\pi$ achieves $\mu$, we have that $\prob_\pi[\Freq{=}\mu] = 1$. This implies that $\prob_\pi[\Freq{=}\mu \mid \MRun_{\calA}] = 1$ for \emph{every} \mbox{$\pi$-eligible $\calA$}. Now consider a profile $\pi_{\calA} = (\xi_1',\ldots,\xi_k')$ such that the initial vertex of every $\xi_i'$ is $v^i$, and the strategy $\xi_i'$ behaves like the strategy $\xi_i$ after visiting the vertex $v^i$.
Since the limit frequency vector of a multirun is the same after deleting an arbitrarily long finite prefix, we have that $\pi_{\calA}$ achieves the frequency vector~$\mu$. Also note that if $\pi$ is a MR or FMR profile, then $\pi_{\calA}$ is a profile of the same type. Let $D_1,\ldots,D_m$ be all MECs of $D$.
Since $\pi_{\calA}$ can be seen as a profile for the MDP $\bigcup_{q=1}^m D_q$ in normal form, we can safely assume that the MDP on input is in normal form.
\section{Experimental Evaluation Details}
\label{sec:experiments-details}

\subsection{Benchmarks}

The plots on Figures~\ref{fig:histogram-size}, \ref{fig:histogram-period} and~\ref{fig:histogram-colors} show histograms of some features of the randomly
generated benchmarks that are used in the experimental evaluation.

\begin{figure}[htb]
  \centering
  \includegraphics[width=\linewidth]{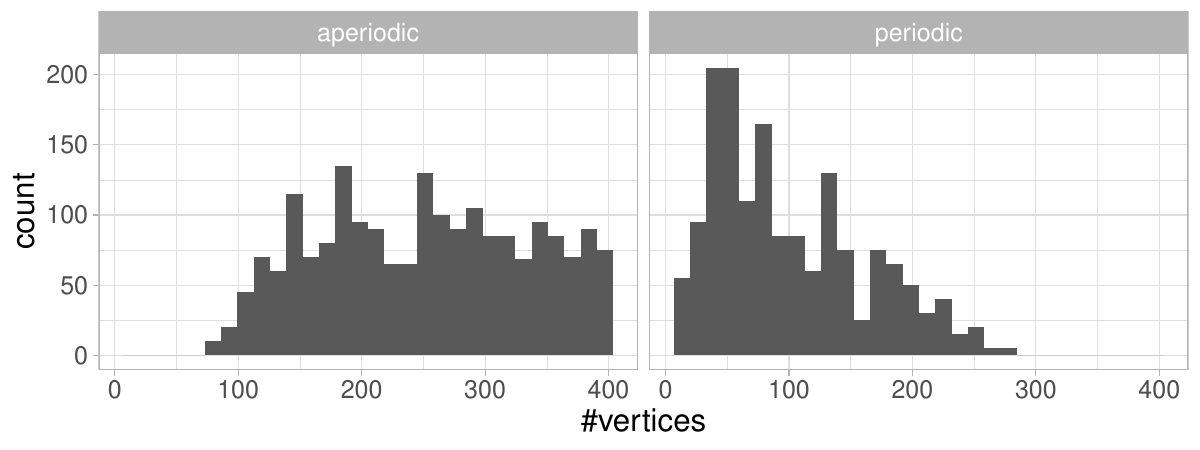}
  \caption{Number of vertices of the generated graphs.}
  \label{fig:histogram-size}
\end{figure}

\begin{figure}[htb]
  \centering
  \includegraphics[width=\linewidth]{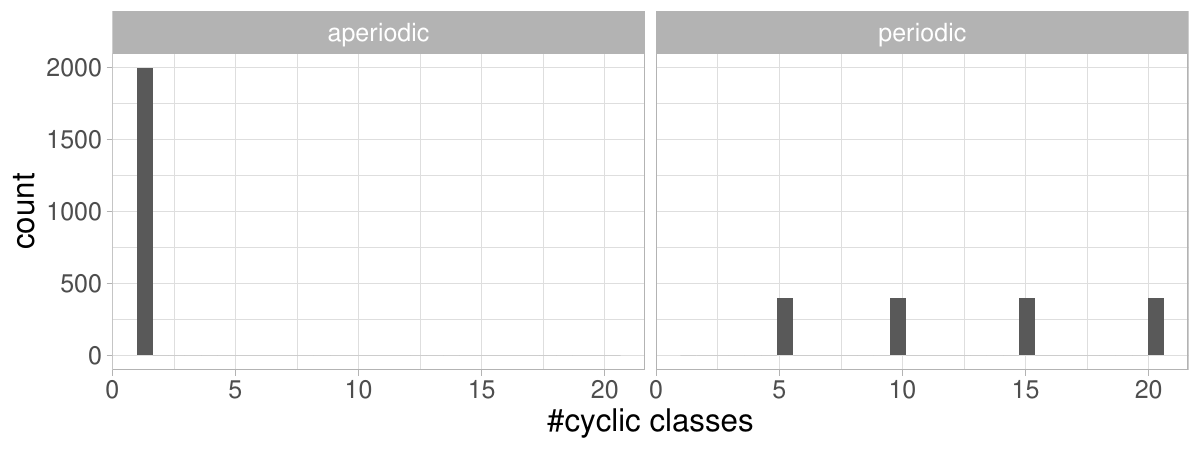}
  \caption{Number of cyclic classes of the generated graphs.}
  \label{fig:histogram-period}  
\end{figure}

\begin{figure}[htb]
  \centering
  \includegraphics[width=\linewidth]{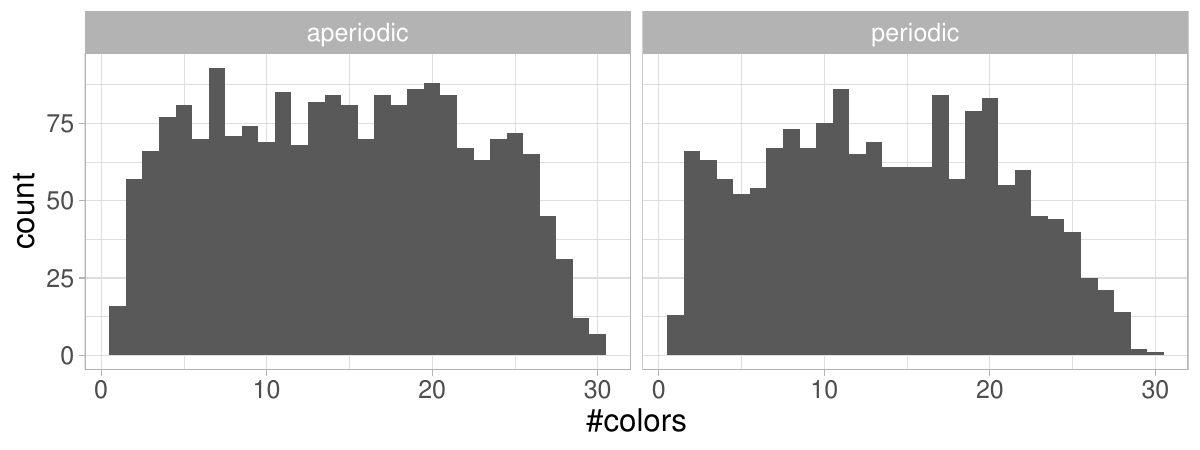}
  \caption{Number of colors of the generated objectives.}
  \label{fig:histogram-colors}  
\end{figure}

\subsection{Running times}

The plots on Figures~\ref{fig:scatter-times} show more details about the running
times of the two algorithms on all benchmarks. All the times are wall times.

\begin{figure}[htb]
  \centering
  \includegraphics[width=\linewidth]{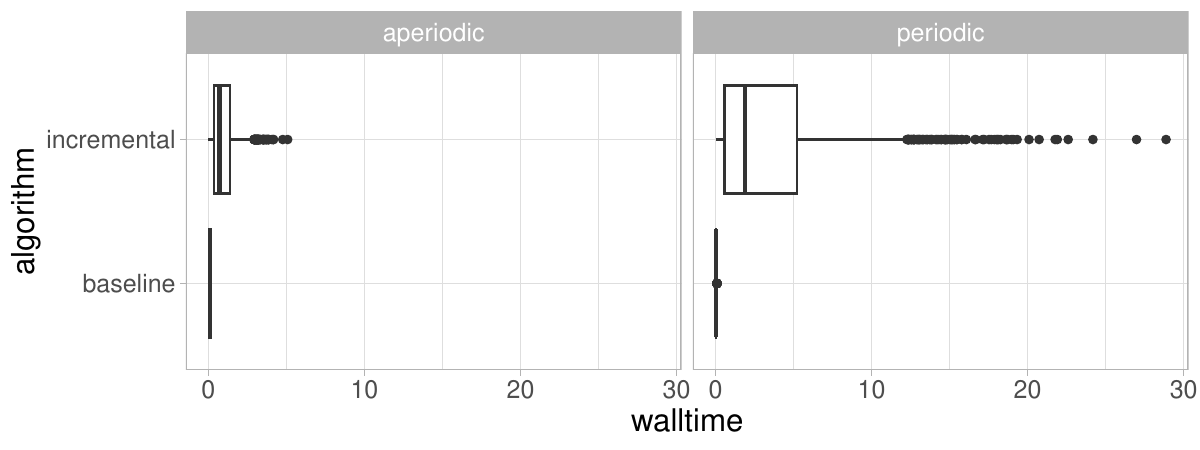}
  \caption{Box plot of running times for both of the algorithms divided by the
    type of the graph (aperiodic/periodic).}
  \label{fig:scatter-times}
\end{figure}

\begin{figure}[htb]
  \centering
  \includegraphics[width=\linewidth]{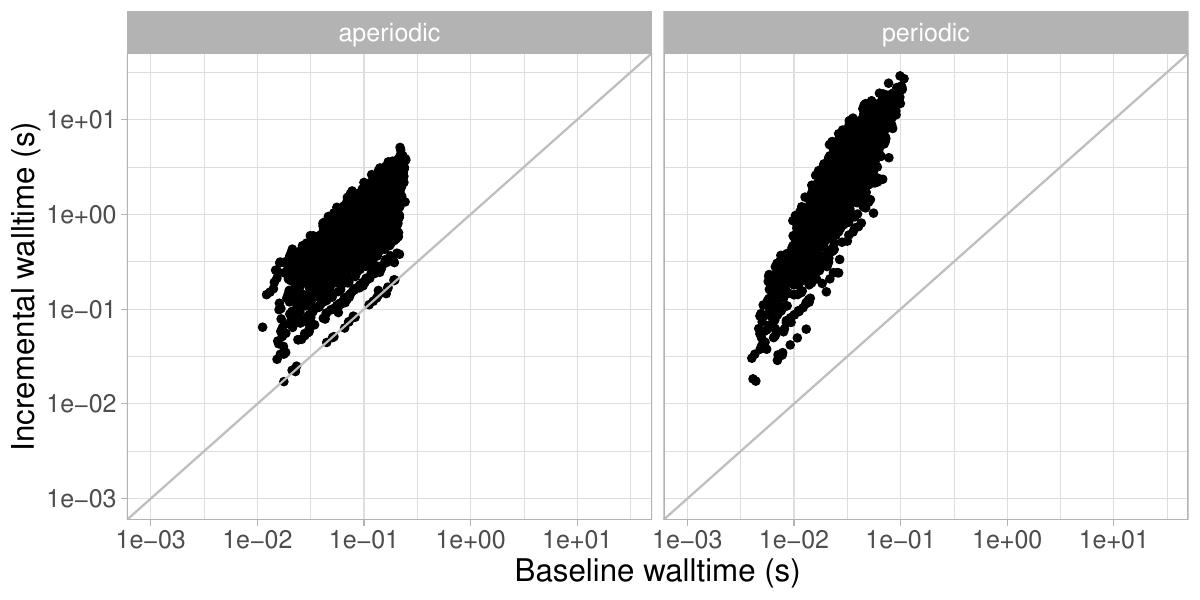}
  \caption{Comparison of running times for all benchmarks. Each dot $(x, y)$ is
    a single benchmark for which the wall time of baseline algorithm was $x$
    seconds and of our incremental algorithm $y$ seconds. The plot is in
    logscale.}
  \label{fig:scatter-times}
\end{figure}

\subsection{Achieved distances}

The plots on Figure~\ref{fig:distances_during_computation_point} show a
variant of Figure 6 from the main paper. The difference is that compared to
Figure 6 of the paper, Figure~\ref{fig:distances_during_computation_point} here
shows \emph{all of the benchmarks}, not a random subset. The values for each
benchmark and a number of agents are shown separately and not on lines, to avoid
visual mess.

The plots on Figure~\ref{fig:l_inf_distances_during_computation} and
Figure~\ref{fig:l_inf_distances_during_computation_point} show the same plots as
Figure 6 of the paper and Figure~\ref{fig:distances_during_computation_point} of
this supplementary material, but are using a different distance from the
objective. These plots show the ``cropped'' version of the $L_\infty$ distance, which is the maximum 
distance from any unsatisfied color. Note that it is a number between $0$ and
$1$ by definition, and does not need to be normalized.

\begin{figure}[htb]
  \centering
  \includegraphics[width=\linewidth]{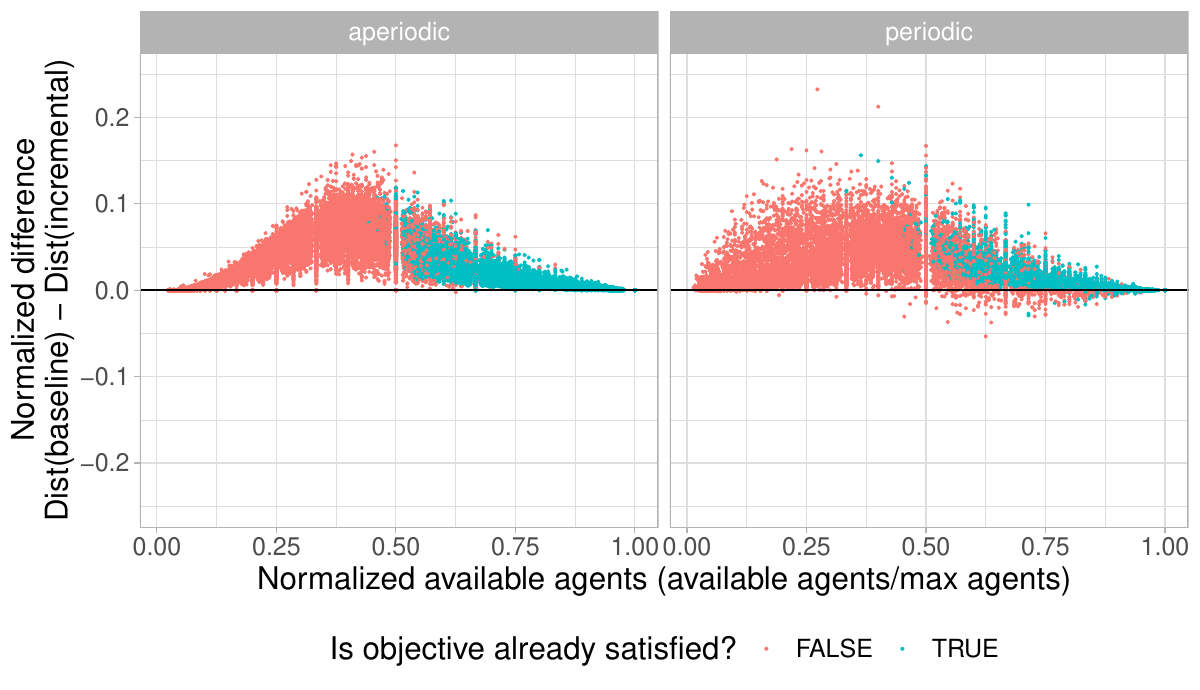}
  \caption{Comparison of distances achieved by the two algorithms on all
    benchmarks. The plot shows the difference of the normalized distances
    $\frac{\Dist(\pi_{\textrm{baseline}}, \Obj)}{|\cols|} -
    \frac{\Dist(\pi_{\textrm{incremental}}, \Obj)}{|\cols|}$ on $y$ axis for
    each number of agents between $0$ and the number sufficient for both of the
    algorithms normalized between $[0,1]$ on $x$ axis. The line is colored blue
    if any of the algorithms has already satisfied the objective.}
  \label{fig:distances_during_computation_point}
\end{figure}

\begin{figure}[htb]
  \centering
  \includegraphics[width=\linewidth]{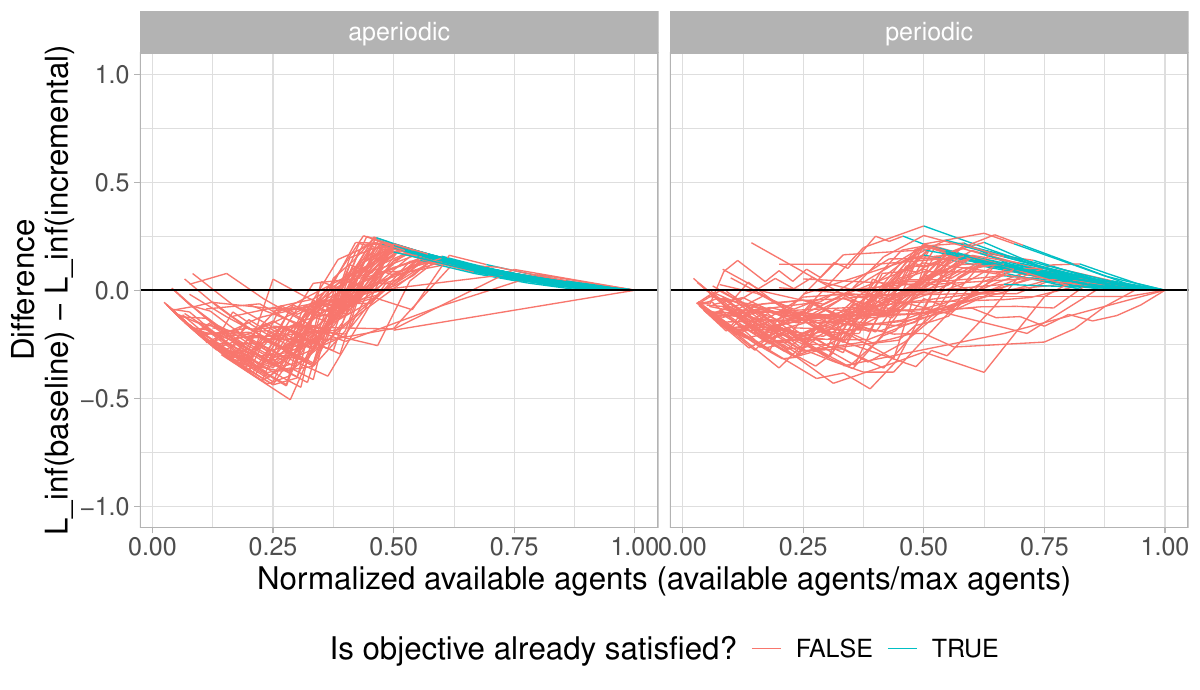}
  \caption{Comparison of $L_\infty$ distances achieved by the two algorithms on a
    randomly selected subset of 150 benchmarks. Each line represents a
    benchmark. The plot shows the difference of the normalized distances
    $L_\infty(\pi_{\textrm{baseline}}, \Obj) - L_\infty(\pi_{\textrm{incremental}},
    \Obj)$ on $y$ axis for each number of agents between $0$ and the number
    sufficient for both of the algorithms normalized between $[0,1]$ on $x$
    axis. The line is colored blue if any of the algorithms has already
    satisfied the objective.}
  \label{fig:l_inf_distances_during_computation}
\end{figure}

\begin{figure}[htb]
  \centering
  \includegraphics[width=\linewidth]{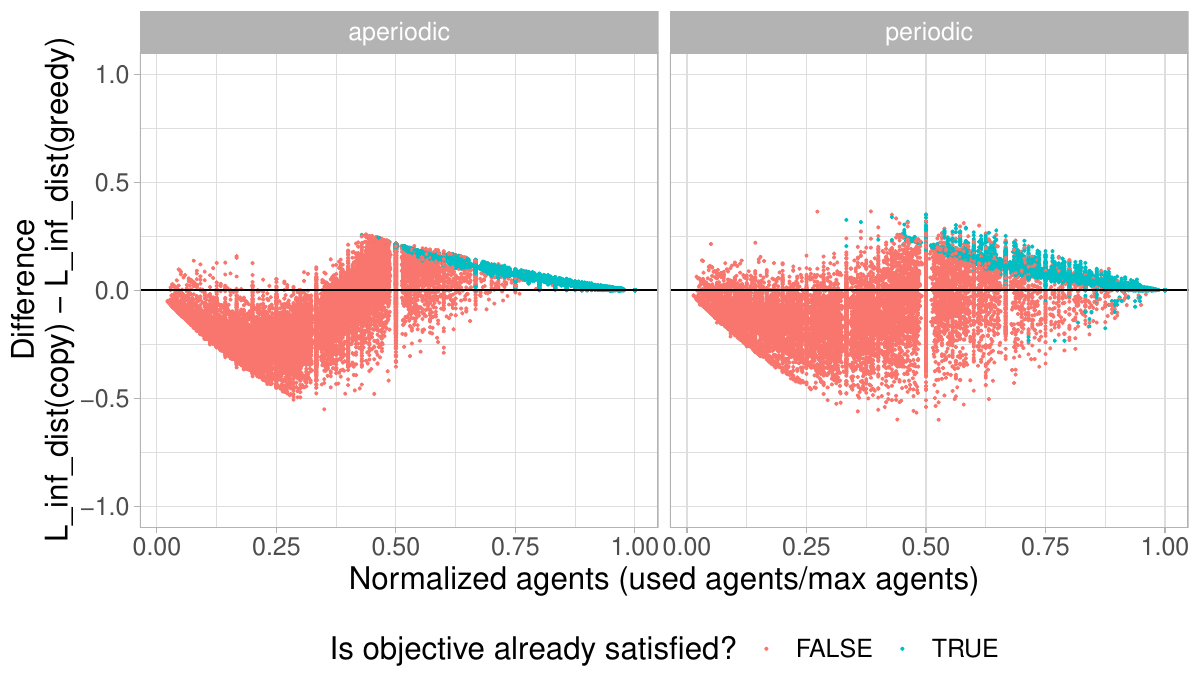}
  \caption{Comparison of $L_\infty$ distances achieved by the two algorithms on
    all benchmarks. The plot shows the difference of the normalized distances
    $L_\infty(\pi_{\textrm{baseline}}, \Obj) -
    L_\infty(\pi_{\textrm{incremental}}, \Obj)$ on $y$ axis for each number of
    agents between $0$ and the number sufficient for both of the algorithms
    normalized between $[0,1]$ on $x$ axis. The line is colored blue if any of
    the algorithms has already satisfied the objective.}
  \label{fig:l_inf_distances_during_computation_point}
\end{figure}

%%% Local Variables:
%%% mode: LaTeX
%%% TeX-master: "suppl"
%%% End:

\end{document}